\documentclass{llncs}

\usepackage{etex}
\usepackage{times}  
\usepackage{url}  
\usepackage{graphicx}  
\usepackage{xspace}
\usepackage{amssymb,amsmath,amsfonts}
\usepackage{autobreak}
\usepackage{mathtools}
\usepackage{temporal}
\usepackage{ifthen,version}
\usepackage{algorithm}
\usepackage[noend]{algpseudocode}
\usepackage{tikz,pgfplots}
\usetikzlibrary{automata,positioning,decorations.markings,arrows,intersections,%
  calc,shapes}
\usepackage{prism}
\usepackage{listings}

\colorlet{darkgreen}{green!40!black}
\colorlet{darkblue}{blue!60!black}
\colorlet{darkred}{red!50!black}
\colorlet{safecellcolor}{yellow!5}
\colorlet{goodcellcolor}{green!5}
\colorlet{badcellcolor}{blue!10}
\definecolor{lightblue}{rgb}{0.5,0.6,1.0}
\definecolor{lightgray}{rgb}{0.95,0.95,0.95}
\definecolor{mauve}{rgb}{0.58,0,0.82}
\definecolor{sienna}{rgb}{0.6,0.18,0.09}

\tikzset{
  >=latex,node distance=2cm,on grid,auto, initial text=,
  box state/.style={draw,rectangle,minimum size=8mm,rounded corners},
  prob state/.style={draw,very thick,shape=circle,darkblue,minimum size=3mm,inner sep=0mm},
  every loop/.style={shorten >=0pt},
  accepting state/.style={double distance=1.2pt, outer sep = 0.6pt+\pgflinewidth},
  accepting dot/.style={above=-2.5pt,circle,fill,darkgreen,inner sep=2pt,radius=1pt},
  loop above/.append style={every loop/.append style={out=120, in=60, looseness=6}},
  loop below/.append style={every loop/.append style={out=300, in=240, looseness=6}},
  loop left/.append style={every loop/.append style={out=210, in=150, looseness=6}},
  loop right/.append style={every loop/.append style={out=30, in=330, looseness=6}}
}

\lstdefinestyle{compact}{basicstyle=\scriptsize,xleftmargin=2ex,numbersep=6pt}
\lstdefinestyle{inline}{basicstyle=\sffamily}

\PassOptionsToPackage{usenames,dvipsnames}{xcolor}
\usepackage{wrapfig}

\usepackage{framed}

\usepackage{paralist}

\usepackage{array}
\usepackage{booktabs}
\usepackage[textwidth=1.5cm,textsize=scriptsize,backgroundcolor=blue!10!white]{todonotes}
\usepackage{wrapfig}
\usepackage{xcolor}
\usepackage{multirow}

\usepackage{listings}

\usepackage{color}

\definecolor{dkgreen}{rgb}{0,0.6,0}
\definecolor{gray}{rgb}{0.5,0.5,0.5}
\definecolor{mauve}{rgb}{0.58,0,0.82}

\newcommand{\pto}{\xrightharpoondown{}}
\newcommand{\set}[1]{\left\{ #1 \right\}}

\DeclareMathOperator{\FRuns}{Paths}
\newcommand{\Runs}{\Omega}

\DeclareMathOperator{\PSat}{\mathsf{PSyn}}
\DeclareMathOperator{\PSemSat}{\mathsf{PSem}}

\newcommand{\f}{{\mu}} 
\newcommand{\Aa}{\mathcal{A}}
\newcommand{\Bb}{\mathcal{B}}
\newcommand{\Mm}{\mathcal{M}}

\newcommand{\Ff}{\mathcal{F}}

\newcommand{\eE}{\mathbb E}

\newcommand{\Real}{\mathbb R}

\newcommand{\DIST}{{\cal D}}

\DeclareMathOperator{\supp}{supp}
\DeclareMathOperator{\last}{last}
\DeclareMathOperator{\infi}{inf}

\title{Good-for-MDPs Automata for Probabilistic Analysis and Reinforcement Learning}
\author{Ernst Moritz Hahn\inst{1,2} \and Mateo Perez\inst{3}
  \and Sven Schewe\inst{4} \and \\  Fabio Somenzi\inst{3} \and
  Ashutosh Trivedi\inst{3} \and Dominik Wojtczak\inst{4}}

\institute{
  School of EEECS, Queen’s University Belfast, UK
  \and
  State Key Laboratory of Computer Science, Institute of Software, CAS, PRC
  \and
  University of Colorado Boulder, USA
  \and
  University of Liverpool, UK
}
\authorrunning{E. M. Hahn, M. Perez, S. Schewe, F. Somenzi, A. Trivedi, and D. Wojtczak}

\begin{document}
\maketitle

\begin{abstract}
  We characterize the class of nondeterministic $\omega$-automata that
  can be used for the analysis of finite Markov decision processes
  (MDPs).  We call these automata `good-for-MDPs' (GFM).  We show that
  GFM automata are closed under classic simulation as well as under
  more powerful simulation relations that leverage properties of
  optimal control strategies for MDPs.  This closure enables us to
  exploit state-space reduction techniques, such as those based on
  direct and delayed simulation, that guarantee simulation
  equivalence.  We demonstrate the promise of GFM automata by defining
  a new class of automata with favorable properties---they are B\"uchi
  automata with low branching degree obtained through a simple
  construction---and show that going beyond limit-deterministic
  automata may significantly benefit reinforcement learning.
\end{abstract}

\section{Introduction}
\label{sec:introduction}

System specifications are often captured in the form of finite
automata over infinite words ($\omega$-automata), which are then used
for model checking, synthesis, and learning.  Of the commonly-used
types of $\omega$-automata, B\"uchi automata have the simplest
acceptance condition, but require nondeterminism to recognize all
$\omega$-regular languages.  Nondeterministic machines can use
unbounded look-ahead to resolve nondeterministic choices.  However,
important applications---like reactive synthesis or model checking and
reinforcement learning (RL) for Markov Decision Process (MDPs
\cite{Put94})---have a game setting, which restrict
the resolution of nondeterminism to be based on the past.

Being forced to resolve nondeterminism on the fly, an automaton may
end up rejecting words it should accept, so that using it can lead to
incorrect results.  Due to this difficulty, initial solutions to these
problems have been based on deterministic automata---usually with
Rabin or parity acceptance conditions.  For two-player games,
Henzinger and Piterman proposed the notion of \emph{good-for-games (GFG)}
automata \cite{Henzin06}.  These are nondeterministic automata that
simulate \cite{Milner71,Henzin97,Etessa05} a deterministic automaton
that recognizes the same language.  The existence of a simulation
strategy means that nondeterministic choices can be resolved without
look-ahead.


The situation is better in the case of probabilistic model checking,
because the game for which a strategy is sought is played on an MDP
against ``blind nature,'' rather than against a strategic opponent who
may take 
advantage of the automaton's inability to resolve
nondeterminism on the fly.  As early as 1985, Vardi noted that
probabilistic model checking can be performed with B\"uchi automata
endowed with a limited form of nondeterminism \cite{Vardi85}. {\em Limit
deterministic B\"uchi automata (LDBA)} \cite{Courco95,Hahn15,Sicker16b} 
perform no nondeterministic choice after seeing 
an accepting transition.  Still, they recognize all $\omega$-regular
languages and are, under mild restrictions \cite{Sicker16b}, \emph{suitable} for
probabilistic model checking.

\noindent\textbf{Related Work.} 
The production of deterministic and limit deterministic automata for model checking has been intensively studied
\cite{Safra89b,Piterm07,BabiakKRS12,ScheweV12,TsaiTH13,TsaiFVT14,ScheweV14,Sicker16b,Duret-LutzLFMRX16,SickertK16,KretinskyMS18},
and several tools are available to produce different types of automata, incl.\ 
MoChiBA/Owl \cite{Sicker16b,SickertK16,KretinskyMS18},
LTL3BA \cite{BabiakKRS12},
GOAL \cite{TsaiTH13,TsaiFVT14},
SPOT \cite{Duret-LutzLFMRX16},
Rabinizer \cite{kvretinsky2018rabinizer},
and B\"uchifier \cite{kini2017optimal}.

So far, only deterministic and a (slightly restricted
\cite{Sicker16b}) class of limit deterministic automata have been
considered for probabilistic model checking \cite{Vardi85,Courco95,Hahn15,Sicker16b}.
Thus, while there have been advances in the efficient production of such automata \cite{Hahn15,Sicker16b,SickertK16,KretinskyMS18}, the consideration of
suitable LDBAs by 
Courcoubetis and Yannakakis in 1988 \cite{CY88}
has been the last time when a fundamental change in the automata foundation of MDP model checking has occurred.

\noindent\textbf{Contribution.}
The simple but effective observation that simulation preserves the suitability for MDPs (for both traditional simulation and the
AEC simulation we introduce) extends the class of automata that can be used in the analysis of MDPs.
This provides us with three advantages:
The first advantage is that we can now use a wealth of simulation based statespace reduction techniques \cite{Dill91,Somenz00,Gurumu02,Etessa05} on an automaton $\cal A$ (e.g.\ an SLDBA)
that we would otherwise use for MDP model checking.
The second advantage is that we can use $\cal A$ to check if a different language equivalent automaton, such as an NBA $\cal B$ (e.g.\ an NBA from which $\cal A$ is derived)
simulates $\cal A$.
For this second advantage, we can dip into the more powerful class of AEC simulation we define in Section \ref{ssec:aec.simulate}
that use properties of winning strategies on finite MDPs.
While this is not a complete method for identifying GFM automata, our experimental results indicate that the GFM property is quite frequent for 
NBAs constructed from random formulas, and can often be established efficiently, while providing a significant statespace reduction
and thus offering a significant advantage for model checking.

A third advantage is that we can use the additional flexibility to tailor automata for different applications than model checking,
for which specialized automata classes have not yet been developed.
We demonstrate this for model-free reinforcement learning (RL).
We argue that RL benefits from three propoerties that are less important in model checking:
The first---easy to measure---property is a small number of successors,
the second and third, are \emph{cautiousness}, the scope for making wrong decisions, and \emph{forgiveness},
the resilience against making wrong decisions, respectively.

A small number of successors is a simple and natural goal for RL, as the lack of an explicit model means that the product space
of a model and an automaton cannot be evaluated backwards. 
In a forward analysis, it matters that nondeterministic choices have to be modeled by enriching the decisions in the MDPs with
the choices made by the automaton.
For LDBAs constructed from NBAs, this means guessing a suitable subset of the reachable states when progressing to the deterministic part of the automaton,
meaning a number of choices that is exponential in the NBA.
We show that we can instead use \emph{slim automata} in Section \ref{ssec:slim} as a first example of
NBAs that are good-for-MDPs, but not limit deterministic.
They have the appealing property that their branching degree is at most two, while keeping the B\"uchi acceptance mechanism that works well with RL \cite{Hahn19}.
(Slim automata can also be used for model checking, but they don't provide similar advantages over suitable LDBAs there, because the backwards analysis
used in model checking makes selecting the correct successor trivial.)

Cautiousness and forgiveness are further properties, which are---while harder to quantify---very desirable for RL:
LDBAs, for example, suffer from having to make a \emph{correct} choice when moving into the deterministic part of the automaton, and
they have to make this correct choice from a very large set of nondeterministic transitions.
While this is unproblematic for standard model checking algorithms that are based on backwards analysis,
applications like RL that rely on forward analysis
can be badly affected when more (wrong) choices are offered, and when wrong choices cannot be rectified.
Cautiousness and forgiveness are a references to this: an automaton is more \emph{cautious} if it has less scope for making wrong decisions
and more \emph{forgiving} if it allows for correcting previously made decisions
(cf.\ Figure \ref{fig:forgiveness} for an example).
%
%
%
%
Our experiments (cf.\ Section \ref{sec:experiments}) indicate that cautiousness and forgiveness are beneficial for RL. 

\noindent\textbf{Organization of the Paper.}
After the preliminaries, we introduce the ``good-for-MDP'' property (Section \ref{sec:slim}) and show that it is preserved by simulation,
which enables all minimization techniques that offer the simulation property (Section \ref{ssec:simulate}).
In Section \ref{ssec:slim} we use this observation to construct slim
automata---NBAs with a branching degree of 2 that are neither limit
deterministic nor good-for-games---as an example of a class of automata that becomes available for MDP model checking and RL.
We then introduce a more powerful simulation relation, AEC simulation, that suffices to establish that an automaton is good-for-MDPs (Section \ref{ssec:aec.simulate}).
In Section \ref{sec:experiments}, we evaluate the impact of the
contributions of the paper on model checking and reinforcement
learning algorithms.


\section{Preliminaries}
\label{sec:problem}
A \emph{nondeterministic B\"uchi automaton} is a tuple
${\cal A} = \langle \Sigma,Q,q_0,\Delta,\Gamma \rangle$, where
$\Sigma$ is a finite \emph{alphabet}, $Q$ is a finite set of
\emph{states}, $q_0 \in Q$ is the \emph{initial state},
$\Delta \subseteq Q \times \Sigma \times Q$ are transitions, and
$\Gamma \subseteq Q \times \Sigma \times Q$ is the transition-based
\emph{acceptance condition}.

A \emph{run} $r$ of ${\cal A}$ on $w \in \Sigma^\omega$ is an
$\omega$-word $r_0, w_0, r_1, w_1, \ldots$ in
$(Q \times \Sigma)^\omega$ such that $r_0 = q_0$ and, for $i > 0$, it
is $(r_{i-1},w_{i-1},r_i)\in \Delta$.  We write $\infi(r)$ for the set
of transitions that appear infinitely often in the run $r$.  A run $r$
of ${\cal A}$ is \emph{accepting} if
$\infi(r) \cap \Gamma \neq \emptyset$.

The \emph{language}, $L_{\mathcal{A}}$, of ${\cal A}$ (or,
\emph{recognized} by ${\cal A}$) is the subset of words in
$\Sigma^\omega$ that have accepting runs in ${\cal A}$.  A language is
$\omega$-\emph{regular} if it is accepted by a B\"uchi automaton.  An
automaton ${\cal A} = \langle\Sigma,Q,Q_0,\Delta,\Gamma\rangle$ is
\emph{deterministic} if $(q,\sigma,q'),(q,\sigma,q'') \in \Delta$
implies $q'=q''$.  ${\cal A}$ is \emph{complete} if, for all
$\sigma \in \Sigma$ and all $q \in Q$, there is a transition
$(q,\sigma,q')\in \Delta$.  A word in $\Sigma^\omega$ has exactly one
run in a deterministic, complete automaton.

A \emph{Markov decision process (MDP)} $\Mm$ is a tuple
$(S, A, T, \Sigma, L)$ where $S$ is a finite set of states, $A$ is a
finite set of \emph{actions}, $T: S \times A \to \DIST(S)$, where
$\DIST(S)$ is the set of probability distributions over $S$, is the
\emph{probabilistic transition function}, $\Sigma$ is an alphabet, and
$L: S \times A \times S \to \Sigma$ is the \emph{labeling function} of
the set of transitions.  For a state $s \in S$, $A(s)$ denotes the set
of actions available in $s$.  For states $s, s' \in S$ and
$a \in A(s)$, we have that $T(s, a)(s')$ equals $\Pr{}(s' | s, a)$.

A \emph{run} of $\Mm$ is an $\omega$-word
$s_0, a_1, \ldots \in S \times (A \times S)^\omega$ such that
$\Pr{}(s_{i+1} | s_{i}, a_{i+1}) > 0$ for all $i \geq 0$.  A finite
run is a finite such sequence.  For a \emph{run}
$r = s_0, a_1, s_1, \ldots$ we define the corresponding labeled
run as
$L(r) = L(s_0,a_1,s_1), L(s_1,a_2,s_2), \ldots \in
\Sigma^\omega$.  We write $\Runs(\Mm)$ ($\FRuns(\Mm)$) for the set of
runs (finite runs) of $\Mm$ and $\Runs_s(\Mm)$ ($\FRuns_s(\Mm)$) for
the set of runs (finite runs) of $\Mm$ starting from state $s$.  When
the MDP is clear from the context we drop the argument $\Mm$.

A strategy in $\Mm$ is a function $\f : \FRuns \to \DIST(A)$ such that
$\supp(\f(r)) \subseteq A(\last(r))$, where $\supp(d)$ is the support
of $d$ and $\last(r)$ is the last state of $r$.  Let $\Runs^\Mm_\f(s)$
denote the subset of runs $\Runs^\Mm(s)$ that correspond to strategy
$\f$ and initial state $s$.  Let $\Sigma_\Mm$ be the set of all
strategies.  We say that a strategy $\f$ is \emph{pure} if $\f(r)$ is
a point distribution for all runs $r \in \FRuns$ and we say that
$\f$ is \emph{positional} if $\last(r) = \last(r')$ implies
$\f(r) = \f(r')$ for all runs $r, r' \in \FRuns$.

The behavior of an MDP $\Mm$ under a strategy $\f$ with starting state
$s$ is defined on a probability space
$(\Runs^\f_s, \Ff^\f_s, \Pr^\f_s)$ over the set of infinite runs of
$\f$ from $s$.  Given a random variable over the set of infinite runs
$f :\Runs \to \Real$, we write $\eE^\f_s \set{f}$ for the expectation
of $f$ over the runs of $\Mm$ from state $s$ that follow strategy
$\f$.


\section{Good-for-MDP (GFM) Automata}
\label{sec:slim}

Given an MDP $\Mm$ and an automaton
$\Aa = \langle \Sigma, Q, q_0, \Delta, \Gamma \rangle$, we want to
compute an optimal strategy satisfying the objective that the run of
$\Mm$ is in the language of $\Aa$.  We define the semantic
satisfaction probability for $\Aa$ and a strategy $\f$ from state $s$
as:
\begin{align*}
\PSemSat^\Mm_{\Aa}(s, \f) &= \Pr{}_s^\f \{ r {\in} \Runs^\f_s :
  L(r) {\in} L_{\mathcal{A}} \} \text{ and} &
  \PSemSat^\Mm_{\Aa}(s) &= \sup_{\f}\big(\PSemSat^\Mm_{\Aa}(s, \f)\big)
  \,.
\end{align*}
When using automata for the analysis of MDPs, we need a syntactic
variant of the acceptance condition.  Given an MDP
$\Mm = (S, A, T, \Sigma, L)$ with initial state $s_0 \in S$
and an automaton
$\mathcal{A} = \langle \Sigma, Q, q_0, \Delta, \Gamma \rangle$, the
\emph{product}
$\Mm \times \mathcal{A} = (S \times Q, (s_0,q_0), A \times Q,
T^\times, \Gamma^\times)$ is an MDP augmented with an initial state $(s_0,q_0)$
and accepting transitions $\Gamma^\times$.  The function
$T^\times : (S \times Q) \times (A \times Q) \pto \DIST(S \times Q)$
is defined by
\begin{equation*}
  T^\times((s,q),(a,q'))(({s}',{q}')) =
  \begin{cases}
    T(s,a)({s}') & \text{if } (q,L(s,a,{s}'),{q}') \in \Delta \\
    0 & \text{otherwise.}
  \end{cases}
\end{equation*}
Finally,
$\Gamma^\times \subseteq (S \times Q) \times (A \times Q) \times (S
\times Q)$ is defined by $((s,q),(a,q'),(s',q')) \in \Gamma^\times$
if, and only if, $(q,L(s,a,s'),q') \in \Gamma$ and $T(s,a)(s') > 0$.
A strategy $\f$ on the MDP defines a strategy $\f^\times$ on the
product, and vice versa.  We define the syntactic satisfaction
probabilities as
\begin{align*}
  \PSat^\Mm_{\Aa}((s,q), \f^\times) &= \Pr{}_s^\f \{ r \in
  \Runs^{\f^\times}_{(s,q)}(\Mm\times\Aa) : \inf(r) \cap \Gamma^\times
  \neq \emptyset \} \enspace, ~~~~\text{ and} \\
  \PSat^\Mm_{\Aa}(s) &= \sup_{\f^\times}\big(\PSat^\Mm_{\Aa}((s,q_0),
  \f^\times)\big) \enspace.
\end{align*}
Note that $\PSat^\Mm_{\Aa}(s) = \PSemSat^\Mm_{\Aa}(s)$ holds for a
deterministic $\Aa$.  In general,
$\PSat^\Mm_{\Aa}(s)$ $\leq \PSemSat^\Mm_{\cal A}(s)$ holds, but equality
is not guaranteed because the optimal resolution of nondeterministic
choices may require access to future events (see Figure
\ref{fig:not-good-for-MDP}).

\begin{definition}[GFM automata]
  \label{def:gfm}
  An automaton $\Aa$ is \emph{good for MDPs} if, for all
  MDPs $\Mm$, $\PSat^\Mm_{\Aa}(s_0) = \PSemSat^\Mm_{\Aa}(s_0)$ holds,
  where $s_0$ is the initial state of $\Mm$.
\end{definition}
For an automaton to match $\PSemSat^\Mm_{\Aa}(s_0)$,
its nondeterminism is restricted not to rely heavily on the future;
rather, it must possible to resolve the nondeterminism on-the-fly.
For example, the B\"uchi automaton presented on the left of Figure \ref{fig:not-good-for-MDP},
which has to guess whether the next symbol is $\mathtt a$ or $\mathtt b$, is not good for MDPs,
because the simple Markov chain on the right of Figure \ref{fig:not-good-for-MDP} does not allow resolution of its nondeterminism on-the-fly.

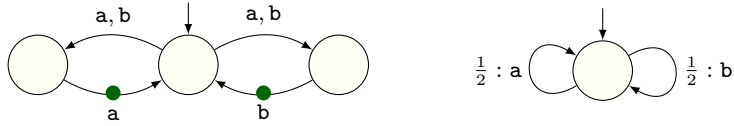
\begin{figure}[t]
  \begin{minipage}{0.59\textwidth}
     \centering
     \begin{tikzpicture}
		\node[state,initial above,fill=safecellcolor] (I) {};
		\node[state,fill=safecellcolor] (A) [left=2cm of I] {};
		\node[state,fill=safecellcolor] (B) [right=2cm of I] {};
		\path[->, bend left]
		(I) edge 			  node[label={[label distance=-2mm]$\mathtt{a,b}$}] {} (B);
		\path[->, bend left]
		(B) edge 			  node [accepting dot,label={[below,label distance=-1.5mm]$\mathtt{b}$}] {} (I);
		\path[->, bend right]
		(I) edge 			  node[label={$\mathtt{a,b}$}] {} (A);
		\path[->, bend right]
		(A) edge 			  node [accepting dot,label={below:$\mathtt{a}$}] {} (I);
	 \end{tikzpicture}
  \end{minipage}
  \begin{minipage}{0.39\textwidth}
	\begin{tikzpicture}
	\node[state,initial above,fill=safecellcolor] (I) {};
	\path[->]
	(I) edge [loop left] node {$\frac{1}{2}: \mathtt{a}$} ();
	\path[->]
	(I) edge [loop right] node {$\frac{1}{2}: \mathtt{b}$} ();
	\end{tikzpicture}
\end{minipage}
  \caption{An NBA, which accepts all words over the alphabet $\{a,b\}$, that is not good for MDPs.
  The dotted transitions are accepting.
  For the Markov chain on the right where
  the probability of $a$ and $b$ is $\frac{1}{2}$,
  the chance that the automaton makes infinitely many correct predictions is $0$}
  \label{fig:not-good-for-MDP}
\end{figure}


There are three classes of automata that are known to be good for MDPs:
(1) deterministic automata,
(2) good for games automata \cite{Henzin06,Klein14}, and
(3) limit deterministic automata that satisfy a few side constraints
\cite{Courco95,Hahn15,Sicker16b}.

A \emph{limit-deterministic} B\"uchi automaton (LDBA) is a
nondeterministic B\"uchi automaton (NBA)
${\cal A} = \langle \Sigma, Q_i \cup Q_f,q_0,\Delta,\Gamma \rangle$ such that
$Q_i \cap Q_f = \emptyset$; $q_0 \in Q_i$;
$\Gamma \subseteq Q_f \times \Sigma \times Q_f$;
$(q,\sigma,q'),(q,\sigma,q'')\in \Delta$ and $q,q'\in Q_f$ implies $q'=q''$; and
$(q,\sigma,q')\in \Delta$ and $q\in Q_f$ implies $q'\in Q_f$.
An LDBA behaves deterministically once it has seen an accepting
transition.
Usual LDBA constructions \cite{Hahn15,Sicker16b} produce GFM automata.
We refer to LDBAs with this property as \emph{suitable} (SLDBAs), cf.\
Theorem \ref{theo:sldbw}.

In the context of RL, techniques based on SLDBAs are particularly
useful, because these automata use the B\"uchi acceptance condition,
which can be translated to reachability goals.  Good for games and
deterministic automata require more complex acceptance conditions,
like parity, that do not have a natural translation into rewards
\cite{Hahn19}.

Using SLDBA~\cite{Courco95,Hahn15,Sicker16b} has the drawback that they naturally have a high branching degree in the initial part,
as they naturally allow for many different transitions to the accepting part of the LDBA.
This can be avoided,
but to the cost of a blow-up and a more complex construction and data structure \cite{Sicker16b}.
We therefore propose an automata construction that produces NBAs with a small branching degree---it never produces more than two successors.
We call these automata \emph{slim}.
The resulting automata are not (normally) limit deterministic, but we show that they are good for MDPs.

Due to technical dependencies we start with presenting a second observation, namely that automata that \emph{simulate} language equivalent GFM automata are GFM.
As a side result, we observe that the same holds for good-for-games automata.
The side result is not surprising, as good-for-games automata were defined through simulation of deterministic automata~\cite{Henzin06}.
But, to the best of our knowledge, the observation from Corollary \ref{cor:reduce} has not been made yet for good-for-games automata.


\subsection{Simulating GFM}
\label{ssec:simulate}
An automaton $\cal A$ \emph{simulates} an automaton $\cal B$ if the duplicator
wins the \emph{simulation game}. 
The simulation game is played between a duplicator and a spoiler, who each
control a pebble, which they move along the edges of $\cal A$ and $\cal B$,
respectively. 
The game is started by the spoiler, who places her pebble on an initial state of $\cal B$.
Next, the duplicator puts his pebble on an initial state of $\cal A$.
The two players then take turns, always starting with the spoiler
choosing an input letter and a transition for that letter in $\cal B$,
followed by the duplicator choosing a transition for the same letter
in $\cal A$.
This way, both players produce an infinite run of their respective automaton.
The duplicator has two ways to win a play of the game: if the run of $\cal A$ he constructs is accepting, and if the run the spoiler constructs on $\cal B$ is rejecting.
The duplicator wins this game if he has a winning strategy, i.e., a
recipe to move his pebble that guarantees that he wins.
Such a winning strategy is ``good-for-games,'' as it can only rely on the past.
It can be used to transform winning strategies of $\cal B$, so that,
if they were witnessing a good for games property or were good for an MDP,
then the resulting strategy for $\mathcal A$ has the same property.

\begin{lemma}[Simulation Properties]
  \label{lem:simple}
  For $\omega$-automata $\cal A$ and $\cal B$ the following holds. 
\begin{enumerate}
 \item If $\cal A$ simulates $\cal B$ then $\cal L(\cal A) \supseteq \cal L(\cal B)$.
 \item If $\cal A$ simulates $\cal B$ and  $\cal L(\cal A) \subseteq \cal L(\cal B)$ then $\cal L(\cal A) = \cal L(\cal B)$.
 \item If $\cal A$ simulates $\cal B$, $\cal L(\cal A) = \cal L(\cal B)$, and $\cal B$ is GFG, then $\cal A$ is GFG.
 \item If $\cal A$ simulates $\cal B$, $\cal L(\cal A) = \cal L(\cal B)$, and $\cal B$ is GFM, then $\cal A$ is GFM.
\end{enumerate} 
\end{lemma}
\begin{proof}
  Facts (1) and (2) are well known observations.  Fact (1) holds
  because an accepting run of $\mathcal B$ on a word $\alpha$ can be
  translated into an accepting run of $\cal A$ on $\alpha$ by using
  the winning strategy of $\mathcal A$ in the simulation game.  Fact
  (2) follows immediately from Fact (1).  Facts (3) and (4) follow by
  simulating the behaviour of $\mathcal B$ on each run. \qed
\end{proof}
This observation allows us to use a family of state-space reduction
techniques, in particular those based on language preserving
translations for B\"uchi automata based on simulation relation
\cite{Dill91,Somenz00,Gurumu02,Etessa05}.  This requires stronger
notions of simulations, like direct and delayed simulation
\cite{Etessa05}.  For the deterministic part of an LDBA, one can also
use space reduction techniques for DBAs like \cite{Schewe10}.

\begin{corollary}
  \label{cor:reduce}
  All statespace reduction techniques that turn an NBA $\cal A$ into
  an NBA $\cal B$ that simulates $\cal A$ preserve GFG and GFM: if
  $\cal A$ is GFG or GFM, then $\cal B$ is GFG or GFM, respectively.
\end{corollary}

\subsection{Constructing Slim GFM Automata}
\label{ssec:slim}
Let us fix B\"uchi automaton $\mathcal{B}=\big\langle\Sigma,Q,Q_0,\Delta,\Gamma\big\rangle$.
We can write $\Delta$ as a function $\hat{\delta}\colon Q\times \Sigma \to 2^Q$ with $\hat{\delta}\colon (q,\sigma) \mapsto \{q'\in Q \mid (q,\sigma,q') \in \Delta \}$,
which can be lifted to sets, using the deterministic transition function $\delta\colon 2^Q \times \Sigma \rightarrow 2^Q$ with $\delta\colon (S,\sigma) \mapsto \bigcup_{q \in S}\hat{\delta}(q,\sigma)$.
We also define an operator, $\mathsf{ndet}$, that translates deterministic transition functions
$\delta\colon R \times \Sigma \rightarrow R$ to relations, using 
$$\mathsf{ndet}\colon (R \times \Sigma \rightarrow R) \rightarrow 2^{R \times \Sigma \times R} \quad \mbox{ with } \quad
\mathsf{ndet}\colon \delta \mapsto \big\{(q,\sigma,q') \mid q' \in \delta(\{q\},\sigma) \big\}.$$
This is just an easy means to move back and forth between functions
and relations, and helps one to visualize the maximal number of successors.
We next define the variations of subset and breakpoint constructions that are
used to define the well-known limit deterministic GFM automata---which we use in
our proofs---and the slim GFM automata we construct. 
Let $3^Q := \big\{(S,S') \mid S' \subsetneq S \subseteq Q \big\}$ and
$3^Q_+ := \big\{(S,S') \mid S' \subseteq S \subseteq Q \big\}$.
We define the subset notation for the transitions and accepting
transitions as $\delta_S,\gamma_S\colon 2^Q \times \Sigma \rightarrow 2^Q$ with 
\begin{align*}
\delta_S&\colon (S,\sigma) {\mapsto} \big\{ q' \in Q \mid \exists q\in
S.\ (q,\sigma,q') \in \Delta\big\}
\text{ and } \\ 
\gamma_S&\colon (S,\sigma) {\mapsto} \big\{ q' \in Q \mid \exists q\in
S.\ (q,\sigma,q') \in \Gamma\big\}.
\end{align*}
We define the raw breakpoint transitions $\delta_R{\colon} 3^Q {\times} \Sigma {\rightarrow} 3^Q_+$ as
$\bigl((S,S'),\sigma\bigr) {\mapsto} \bigl(\delta_S( S,\sigma),$ $\delta_S(S',\sigma)\cup\gamma_S(S,\sigma)\bigr)$.
In this construction, we follow the set of reachable states (first set) and the
states that are reachable while passing at least one of the accepting
transitions (second set).
To turn this into a breakpoint automaton, we reset the second set to the empty
set when it equals the first; the transitions where we reset the second set are
exactly the accepting ones. 
The breakpoint automaton $\mathcal D =
\big\langle\Sigma,3^Q,(Q_0,\emptyset),\delta_B,\gamma_B\big\rangle$ is defined such that,  
when $\delta_R\colon \big((S,S'),\sigma\big) \mapsto  (R,R')$, then there are
three cases:
\begin{enumerate}
\item
  if $R=\emptyset$, then $\delta_B\big((S,S')\big)$ is undefined (or, if a
  complete automaton is preferred, maps to a rejecting sink), 
\item
  else, if $R \neq R'$, then $\delta_B\colon \big((S,S'),\sigma\big) \mapsto  (R,R')$ is a non-accepting transition, 
 \item otherwise $\delta_B,\gamma_B\colon \big((S,S'),\sigma\big) \mapsto  (R',\emptyset)$ is an accepting transition.
\end{enumerate}
Finally, we define transitions $\Delta_{SB} \subseteq 2^Q \times \Sigma \times 3^Q$ that lead from a subset to a breakpoint construction, and $\gamma_{2,1}\colon 3^Q \times \Sigma \rightarrow 3^Q$ that promote the second set of a breakpoint construction to the first set as follows.

\begin{enumerate}
 \item $\Delta_{SB} =  \Big\{\big(S,\sigma,(S',\emptyset)\big) \mid \emptyset \neq S' \subseteq \delta_S(S,\sigma)\Big\}$ are non-accepting transitions,
 \item if $\delta_S(S',\sigma) = \gamma_S(S,\sigma) = \emptyset$, then $\gamma_{2,1}\big((S,S'),\sigma\big)$ is undefined, and
 \item otherwise $\gamma_{2,1}\colon \big((S,S'),\sigma\big) \mapsto  \big(\delta_S(S',\sigma)\cup\gamma_S(S,\sigma),\emptyset\big)$ is an accepting transition.
\end{enumerate}
We can now define standard limit deterministic good for MDP automata.
\begin{theorem}\cite{Hahn15}\label{theo:sldbw}
$\mathcal A = \big\langle\Sigma,2^Q \cup 3^Q,Q_0,\mathsf{ndet}(\delta_S) \cup \Delta_{SB} \cup \mathsf{ndet}(\delta_B),\mathsf{ndet}(\gamma_B)\big\rangle$
recognizes the same language as $\mathcal B$. It is limit deterministic and good for MDPs.
\end{theorem}
We now show how to construct a slim GFM B\"uchi automaton.
\begin{theorem}[Slim GFM B\"uchi Automaton]
\label{theo:slim}
The automaton
\begin{equation*}
\mathcal S =
\bigl\langle\Sigma,3^Q,(Q_0,\emptyset),\mathsf{ndet}(\delta_B) \cup
\mathsf{ndet}(\gamma_{2,1}),\mathsf{ndet}(\gamma_B) \cup
\mathsf{ndet}(\gamma_{2,1})\bigr\rangle
\end{equation*}
simulates $\cal A$.
$\mathcal S$ is slim, language equivalent to $\mathcal B$, and good for MDPs.
\end{theorem}

\begin{proof}

$\cal S$ is slim: its set of transitions is the union of two sets of deterministic transitions.
We show that $\mathcal S$ simulates $\mathcal A$
by defining a strategy in the simulation game, which ensures that,
if the spoiler produces a run $S_0 \ldots S_{j-1} (S_j,S_j') (S_{j+1},S_{j+1}')\ldots$ for $\cal A$, then
the duplicator produces a run $(T_{0},T_{0}')$ $\ldots (T_{j-1},T_{j-1}') (T_{j},T_{j}') (T_{j+1},T_{j-1}')\ldots$ for $\cal S$, such that
(1) $S_i \subseteq T_i$ holds for all $i \in \omega$, and
(2) if there are two accepting transitions $\big((S_{k-1},S_{k-1}'),\sigma_k, (S_{k},S_{k}')\big)$ and  $\big((S_{l-1},S_{l-1}'),\sigma_l, (S_{l},S_{l}')\big)$
with $k<l$, there is an $k < m \leq l$, such that
$\big((T_{m-1},T_{m-1}'),$ $\sigma_m (T_{m},T_{m}')\big)$ is accepting.

To obtain this, we describe a winning strategy for the duplicator while arguing inductively that it mainains (1).
Note that (1) holds initially ($T_0 = S_0$, induction basis).
\\[2mm]
\noindent\textbf{Initial Phase:}
Every move of the spoiler---with some letter $\sigma$---that uses a transition from $\delta_S$---the subset part of $\cal A$---is followed by a move from $\delta_B$ with the same letter $\sigma$.
When the duplicator follows this strategy the following holds:
when, after a pair of moves, the pebble of the spoiler is on state $S \subseteq Q$, then the pebble of the duplicator is on some state $(S,S')$.
In particular, (1) is preserved during this phase (induction step).
\\[2mm]
\noindent\textbf{Transition Phase:}
The one spoiler move---with some letter $\sigma$---that uses a transition from $\Delta_{SB}$---the transition to the breakpoint part of $\cal A$---is followed by a move from $\delta_B$ with the same letter $\sigma$.
When the duplicator follows this strategy, and when, after the pair of moves, the pebble of the spoiler is on state $(S,\emptyset)$, then the pebble of the duplicator is on some state $(T,T')$ with $S\subseteq T$.
In particular, (1) is preserved (induction step). 
\\[2mm]
\noindent\textbf{Final Phase:}
When the spoiler moves from some state $(S,S')$---with some letter $\sigma$---that uses a transition from $\delta_B$---the breakpoint part of $\cal A$---%
to $(\bar{S},\bar{S}')$, and when the duplicator is in some state $(T,T')$, then the duplicator does the following.
He calculates $(\bar{T},\emptyset)=\gamma_{2,1}\big((T,T'),\sigma\big)$ and checks if $\bar{S} \subseteq \bar{T}$ holds.
If  $\bar{S} \subseteq \bar{T}$ holds, he plays this transition from $\gamma_{2,1}$ (with the same letter $\sigma$).
Otherwise, he plays the transition from $\delta_B$ (with the same letter $\sigma$).
In either case (1) is preserved (induction step), which closes the inductive argument for (1).

Note that no accepting transition of $\mathcal A$ is passed in the initial or tansition phase, so the two accepting transitions from (2) must both fall into the final phase.

To show (2), we first observe that $S_k' = \emptyset$, and thus $S_k' \subseteq T_k'$ holds.
Assuming for contradition that all transitions of $\mathcal S$ for $\sigma_{k+1}\ldots\sigma_{l-1}$ are non-accepting,
we obtain---using (1)---by a straightforward inductive argument that $S_i' \subseteq T_i'$ for all $i$ with $k{\leq} i {<} l$.
(Note that transitions in $\delta_B$ are accepting when they are also be in $\gamma_B$.)

Using that $S_l = \delta_S(S_{l-1}',\sigma_l) \cup \gamma_S(S_{l-1},\sigma_l) \subseteq  \delta_S(T_{l-1}',\sigma_l) \cup \gamma_S(T_{l-1},\sigma_l)$ holds,
the spoiler uses an accepting transition from $\gamma_{2,1}$ in this step.


Using Lemma \ref{lem:simple}, it now suffices to show that the language of $\mathcal S$ is included in the language of $\mathcal B$.
To show this, we simply argue that an accepting run $\rho = (Q_0,Q_0') , (Q_1,Q_1') ,$ $ (Q_2,Q_2') , (Q_3,Q_3') , \ldots$ of $\mathcal S$ on an input word $\alpha = \sigma_0,\sigma_1,\sigma_2,\ldots$ can be interpreted as a forest of finitely many finitely branching trees of overall infinite size, where all infinite branches are accepting runs of $\mathcal B$.
K\H{o}nig's Lemma then proves the existence of an accepting run of $\mathcal B$.

This forest is the usual one.
The nodes are labeled by states of $\mathcal B$, and the roots (level 0) are the initial states of $\mathcal B$.
Let $I = \bigl\{ i \in \mathbb N \mid \big((Q_{i-1},Q_{i-1}'),\sigma_{i-1},(Q_{i},Q_{i}')\big) \in \Gamma := \mathsf{ndet}(\gamma_B) \cup \mathsf{ndet}(\gamma_{2,1})\bigr\}$ be the set of positions after accepting transitions in $\rho$.
We define the predecessor function $\mathsf{pred}\colon \mathbb N \rightarrow I \cup \{0\}$ with $\mathsf{pred}\colon i \mapsto \max\big\{j \in I \cup \{0\} \mid j < i\big\}$.

We call a node with label $q_l$ on level $l$ an end-point if one of the following applies:
(1) $q_l \notin Q_l$ or 
(2) $l \in I$ and for all $j$ such that $\mathsf{pred}(l) \leq j < l$, where $q_j$ is the label of the ancestor of this node on level $j$,
we have $(q_j,\sigma_j,q_{j+1}) \notin \Gamma$.

(1) may only happen after a transition from $\gamma_{2,1}$ has been taken, and the $q_l$ is not among the states that is traced henceforth.
(2) identifies parts of the run tree that do not contain an accepting transition.

A node labeled with $q_l$ on level $l$ that is not an endpoint has $\big|\delta_S(q_l,\sigma_l)\big|$ children, labeled with the different elements of $\delta_S(q_l,\sigma_l)$.
It is now easy to show by induction over $i$ that the following holds.
\begin{enumerate}
 \item For all $q \in Q_i$, there is a node on level $i$ labeled with $q$.
 \item For $i \notin I$ and $q \in Q_i'$, there is a node labeled $q$ on level $i$, a $j$ with $\mathsf{pred}(i) \leq j < i$,
 and ancestors on level $j$ and $j+1$ labeled $q_j$ and $q_{j+1}$, such that
$(q_j,\sigma_j,q_{j+1}) \in \Gamma$. (The `ancestor' on level $j+1$ might be the state itself.)

For $i \in I$  and $q \in Q_i'$, there is a node labeled $q$ on level $i$, which is not an end point.
\end{enumerate}
Consequently, the forest is infinite, finitely branching, and finitely
rooted, and thus contains an infinite path.  By construction, this
path is an accepting run of $\cal B$. \qed


\end{proof}

\begin{figure}[t]
  \centering
  \begin{tikzpicture}[every text node part/.style={align=center}]
    \begin{scope}
      \node[state,shape=ellipse,darkgreen,fill=goodcellcolor,inner sep=1pt] (S01) at (4,3) {$\{0,1\}$};
      \coordinate [left=1cm of S01] (ild);
      \node[state,shape=ellipse,fill=badcellcolor,inner sep=1pt] (S1-e) [below left=2cm and 2.75cm of
      S01] {$\{1\}$\\$\emptyset$};
      \node[state,shape=ellipse,fill=badcellcolor,inner sep=1pt] (S0-e) [right=2.75cm
      of S1-e] {$\{0\}$\\$\emptyset$};
      \node[state,shape=ellipse,initial below,fill=safecellcolor,inner sep=1pt] (S01-e) [right=2.75cm of S0-e]
      {$\{0,1\}$\\$\emptyset$}; 
      \node[state,shape=ellipse,fill=safecellcolor,inner sep=1pt] (S01-0) [right=4cm of S01-e]
      {$\{0,1\}$\\$\{0\}$}; 
      \path[->]
      (ild) edge[darkgreen] (S01)
      (S01) edge[darkgreen,loop right] node {$a,b$} ()
      edge[darkgreen,swap] node[sloped,above] {$a,b$} (S1-e)
      edge[darkgreen] node[sloped,above] {$a,b$} (S0-e)
      edge[darkgreen] node[sloped,above] {$a,b$} (S01-e)
      (S1-e) edge node[accepting dot, label={$a$}] {} (S0-e)
      (S0-e) edge node {$a,b$} (S01-e)
      (S01-e) edge[bend left=30] node {$a$} (S01-0)
      edge [loop above, looseness=4, out=120, in=60] node {$b$} (S01-e)
      (S01-0) edge node[accepting dot, label={$a,b$}] {} (S01-e)
      (S01-0) edge[red,dashed,out=210,in=-30] node[accepting
      dot,label={$a,b$}] {} (S01-e);
    \end{scope}
    \begin{scope}[xshift=10cm,yshift=2.75cm]
      \node[state,shape=ellipse,initial above,fill=safecellcolor] (S0) at (0,0) {$0$};
      \node[state,shape=ellipse,initial above,fill=safecellcolor] (S1) [right=2.25cm of S0] {$1$};
      \path[->]
      (S0) edge[bend left] node {$a,b$} (S1)
      edge[loop left] node {$a,b$} ()
      (S1) edge[bend left] node[accepting dot, label={$a$}] {} (S0)
      ;
    \end{scope}
  \end{tikzpicture}
  \caption{An NBA for $\always\eventually a$ (in the upper right
    corner) together with an SLDBA and a slim NBA constructed
    from it.  The SLDBA and the slim NBA are shown sharing their common
    part.  State $\{0,1\}$, produced by the subset construction, is
    the initial state of the SLDBA, while state
    $(\{0,1\},\emptyset)$---the initial state of the breakpoint
    construction---is the initial state of the slim NBA.  States
    $(\{1\},\emptyset)$ and $(\{0\},\emptyset)$ are states of the
    breakpoint construction that only belong to the SLDBA because they
    are not reachable from $(\{0,1\},\emptyset)$.  The transitions out
    of $\{0,1\}$, except the self loop, belong to $\Delta_{SB}$.  The
    dashed-line transition from $(\{0,1\},\{0\})$ belongs to
    $\gamma_{2,1}$}
  \label{fig:slim-GFa}
\end{figure}
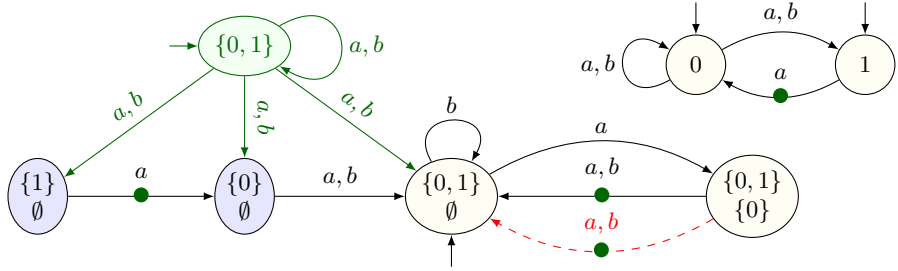


The resulting automata are simple in structure and enable symbolic implementation (See Fig.~\ref{fig:slim-GFa}).
It cannot be expected that there are much smaller good for MDP automata, as its
explicit construction is the only non-polynomial part in model checking MDPs. 

\begin{theorem}
Constructing a GFM B\"uchi automaton $G$ that recognizes the models of an LTL formula $\varphi$ requires time doubly exponential in $\varphi$, and
constructing a GFM B\"uchi automaton $G$ that recognizes the language of an NBA $\mathcal B$ requires time exponential in $\mathcal B$.
\end{theorem}

\begin{proof}
As resulting automata are GFM, they can be used to model check MDPs $\cal M$
against this property, with cost polynomial in product of $\cal M$ and $\cal
G$. 
If $\cal G$ could be produced faster (and if they could, consequently be
smaller) than claimed, it will contradict the 2-{\sc ExpTime-} and {\sc
  ExpTime}-hardness \cite{Courco95} of these model checking problems. \qed
\end{proof}

%

\section{Accepting End-Component Simulation}
\label{ssec:aec.simulate}

An \emph{end-component}~\cite{deAlfa98,Baier08} of an MDP $\Mm$ is a
sub-MDP $\Mm'$ of $\Mm$ such that its underlying graph is strongly
connected.  A \emph{maximal} end-component is maximal under set-inclusion.
Every state of an MDP belongs to at most one maximal end-component.

\begin{theorem}[End-Component Properties.  Theorem 3.1 and Theorem 4.2
  of~\cite{deAlfa98}]
  \label{th:end-comp}
  Once an end-component $C$ of an MDP is entered, there is a strategy
  that visits every state-action combination in $C$ infinitely often with probability
  $1$ and stays in $C$ forever.
  
  For a product MDP, an \emph{accepting end-component} (AEC) is an
  end-component that contains some transition in $\Gamma^\times$.
  There is a positional pure strategy for an AEC $C$ that surely stays
  in $C$ and almost surely visits a transition in $\Gamma^\times$
  infinitely often.
  
  For a product MDP, there is a set of disjoint accepting
  end-components such that, from every state, the maximal probability
  to reach the union of these accepting end-components is the same as
  the maximal probability to satisfy $\Gamma^\times$.  Moreover, this
  probability can be realized by combining a positional pure
  (reachability) strategy outside of this union with the
  aforementioned positional pure strategies for the individual AECs.
\end{theorem}

Lemma~\ref{lem:simple} shows that the GFM property is preserved by
simulation: For language-equivalent automata $\Aa$ and $\Bb$, if $\Aa$
simulates $\Bb$ and $\Bb$ is GFM, then $\Aa$ is also GFM.  However, a
GFM automaton may not simulate a language-equivalent GFM automaton.
(See Figure~\ref{fig:AECexample}.)  Therefore we introduce a coarser
preorder, Accepting End-Component (AEC) simulation, that exploits the
finiteness of the MDP $\Mm$.  We rely on Theorem \ref{th:end-comp} to
focus on positional pure strategies for $\Mm \times \Bb$.  Under such
strategies, $\Mm \times \Bb$ becomes a Markov chain \cite{Baier08}
such that almost all its runs have the following properties:
\begin{itemize}
 \item They will eventually reach a leaf strongly connected component (LSCC) in the Markov chain.
 \item If they have reached a LSCC $L$, then, for all $\ell \in \mathbb N$,
 all sequences of transitions of length $\ell$ in $L$ occur infinitely often,
 and no other sequence of length $\ell$ occurs.
\end{itemize}
With this in mind, we can intuitively ask the spoiler to pick a run
through this Markov chain, and to disclose information about this run.
Specifically, we can ask her to signal when she has reached an
accepting LSCC%
\footnote{There is nothing to show when a non-accepting LSCC is
  reached---if $\Bb$ rejects, then $\Aa$ may reject too---nor when no
  LSCC is reached, as this occurs with probability $0$.}  in the
Markov chain, and to provide information about this LSCC, in
particular information entailed by the full list of sequences of
transitions of some fixed length $\ell$ described above.  Runs that
can be identified to either not reach an accepting LSCC, to visit
transitions not in this list, or to visit only a subset of sequences
from this list, form a $0$ set.  In the simulation game we define
below, we make use of this observation to discard such runs.

A simulation game can only use the syntactic material of the
automata----neither the MDP nor the strategy are available.  The
information the spoiler may provide cannot explicitly refer to them.
What the spoiler may be asked to provide is information on when she
has entered an accepting LSCC, and, once she has signaled this, which
sequences of length $l$ of \emph{automata} transitions of $\Bb$ occur
in the LSCC.  The sequences of automata transitions are simply the
projections on the automata transitions from the sequences of
transitions of length $\ell$ that occur in the LSCC $L$.  We call this
information a \emph{gold-brim accepting end-component claim} of length
$\ell$, $\ell$-GAEC claim for short.

The term ``gold-brim'' in the definition indicates that this is a powerful approach, but not one that can be implemented efficiently.
We will define weaker, efficiently implementable notions of accepting end-component claims (AEC claims) later.

The AEC simulation game is very similar to the simulation game of
Section~\ref{ssec:simulate}.  Both players produce an infinite run of
their respective automata.  If the spoiler makes an AEC claim, e.g., an $\ell$-GAEC claim, we say
that her run \emph{complies} with it if, starting with the transition when the
AEC claim is made, all states, transitions, or sequences of
transitions in the claim appear infinitely often, and all states,
transitions, and sequences of transitions the claim excludes do not
appear.
For an $\ell$-GAEC claim, this means that all of the sequences of
transitions of length $\ell$ in the claim occur infinitely often, and
no other sequence of length $\ell$ occurs henceforth.

Thus, like a classic simulation game, an $\ell$-GAEC simulation game is started by the spoiler, who places her pebble on an initial state of $\cal B$.
Next, the duplicator puts his pebble on an initial state of $\cal A$.
The two players then take turns, always starting with the spoiler choosing an input letter and an according transition from $\cal B$,
followed by the duplicator choosing a transition for the same letter in $\cal A$.

Different from the classic simulation game, in an $\ell$-GAEC
simulation game, the spoiler has an additional move that she can (and,
in order to win, has to) perform once in the game: In addition to
choosing a letter and a transition, 
she can claim that she has reached an accepting end-component, and provide a complete list of sequences of automata transitions of length $\ell$ that
can henceforth occur.
This store is maintained, and never updated. It has no further effect on the rules of the game:
Both players produce an infinite run of their respective automata.
The duplicator has four ways to win: 
\begin{enumerate}
 \item if the spoiler never makes an AEC claim,
 \item if the run of $\Aa$ he constructs is accepting,
 \item if the run the spoiler constructs on $\Bb$ does not comply with the AEC claim, and
 \item if the run that the spoiler produces is not accepting.
\end{enumerate}
For $\ell$-GAEC claims, (4) simply means that the set of transitions
defined by the sequences does not satisfy the B\"uchi, parity, or
Rabin acceptance condition.

\begin{theorem}\label{theo:GAEC}
[$\ell$-GAEC Simulation]
If $\cal A$ and $\cal B$ are language equivalent automata, $\cal B$ is
GFM, and there exists an $\ell$ such that $\cal A$ $\ell$-GAEC
simulates $\cal B$, then $\cal A$ is GFM.
\end{theorem}


For the proof, we use an arbitrary (but fixed) MDP $\cal M$, and an arbitrary (but fixed) pure optimal positional strategy $\f$ for
$\cal M \times \cal B$, resulting in the Markov chain $(\cal M \times \cal B)_\f$.
We assume w.l.o.g.\ that the accepting LSCCs in $(\cal M \times \cal B)_\f$ are identified, e.g., by a bit.

Let $\tau$ be a winning strategy of the duplicator
in an $\ell$-GAEC simulation game.  Abusing notation, we let
$\tau \circ \f$ denote the finite-memory strategy%
\footnote{The strategy $\tau$ consists of one sub-strategy to be used
  before the AEC claim is made and one sub-strategy for each possible
  $\ell$-GAEC claim.  The memory of $\tau \circ \f$ tracks the
  position in $(\Mm \times \Bb)_\f$.  When an accepting LSCC is
  detected (via the marker bit) analysis of $(\Mm \times \Bb)_\f$
  reveals the only possible $\ell$-GAEC claim.  This claim is used to
  select the right entry from $\tau$.}  obtained from $\f$ and $\tau$
for $\cal M \times \cal A$, where $\tau$ is acting only on the
automata part of $(\cal M \times \cal B)$, and where the spoiler makes
the move to the end-component when she is in some LSCC $B$ of
$(\cal M \times \cal B)_\f$ and 
gives the full list of sequences of transitions of length $\ell$ that
occur in $B$.

\begin{proof}
As $\cal B$ is good for MDPs, we only have to show that the chance of winning in
$(\cal M \times \cal A)_{\tau \circ \f}$ is at least the chance of winning in
$(\cal M \times \cal B)_\f$. 
The chance of winning in $(\cal M \times \cal B)_\f$ is the chance of reaching
an accepting LSCC in $(\cal M \times \cal B)_\f$. 
It is also the chance of reaching an accepting LSCC $L \in (\cal M \times \cal
B)_\f$ \emph{and}, after reaching $L$, to see exactly the sequences of
transitions of length $\ell$ that occur in $L$, and to see all of them infinitely
often.  

By construction, $\tau \circ \f$ will translate those runs into accepting runs
of $(\cal M \times \cal A)_{\tau \circ \f}$, such that the chance of an
accepting run of $(\cal M \times \cal A)_{\tau \circ \f}$ is at least the chance
of an accepting run of $(\cal M \times \cal B)_\f$.  As $\f$ is optimal, the
chance of winning in  $\cal M \times \cal A$ is at least the chance of winning
in $\cal M \times \cal B$.
As $\cal B$ is GFM, this is the chance of $\cal M$ producing a run
accepted by $\cal B$ (and thus $\cal A)$ when controlled optimally, which is an
upper bound on the chance of winning in  $\cal M \times \cal A$. \qed
\end{proof}




An $\ell$-GAEC simulation, especially for large $\ell$, results in
very large state spaces, because the spoiler has to
list all sequences of transitions of $\Bb$ of length $\ell$
that will appear infinitely often.  No other sequence of
length $\ell$ may then appear in the run%
\footnote{The AEC claim provides information about the
  accepting LSCC in the product under the chosen pure
  positional strategy.  When the AEC claim requires the exclusion of
  states, transitions, or sequences of transitions, then they are
  therefore surely excluded, whereas when it requires inclusion of,
  and thus inclusion of infinitely many occurrances of, states,
  trasitions, or sequences of transitions, then they (only) occur almost
  surely infinitely often. Yet, runs that do not contain them all
  infinitely often form a zero set, and can thus be ignored.}.
This can, of course, be prohibitively expensive.

As a compromise, one can use coarser-grained information at the cost
of reducing the duplicator's ability of winning the game. E.g., the
spoiler could  be asked to only reveal a transition that is
repeated infinitely often, plus (when using more powerful acceptance conditions than B\"uchi),
some acceptance information, say the dominating priority in a parity game or a winning Rabin pair.
This type of coarse-grained claim can be refined slightly by allowing
the \emph{duplicator} to change at any time the transition that is to
appear infinitely often to the transition just used by the spoiler.
%
%
Generally, we say that an AEC simulation game is any simulation game, where
\begin{itemize}
 \item the spoiler provides a list of states, transitions, or sequences of transitions that will occur infinitely often and
       a list of states, transitions, or sequences of transitions that will not occur in the future when making her AEC claim, and
 \item the duplicator may be able to update this list based on his observations,
 \item there exists some $\ell$-GAEC simulation game such that a
   winning strategy of the spoiler translates into a winning strategy
   of the spoiler in the AEC simulation game.
\end{itemize}
The requirement that a winning spoiler strategy translates into a winning spoiler strategy in an $\ell$-GAEC game entails that AEC
simulation games can prove the GFM property.

\begin{corollary}\label{cor:AECsim}
[AEC Simulation]
If $\cal A$ and $\cal B$ are language equivalent automata, $\cal B$ is good for MDPs, and
$\cal A$ AEC-simulates $\cal B$, then $\cal A$ is good for MDPs.
\end{corollary}
Of course, for every AEC simulation, one first has to prove that winning strategies for the spoiler translate.
We have used two simple variations of the AEC simulation games:\\[-2mm]

\noindent\textbf{accepting transition:}
the spoiler may only make her AEC claim when taking an accepting transition; this transition---and no other information---is stored, and
the spoiler commits to---and commits only to---seeing this transition
infinitely often;
\\[2mm] 
\noindent\textbf{accepting transition with update:}
different to the \emph{accepting transition} AEC simulation game, the duplicator can---but does not have to---update the stored accepting transition
whenever the spoiler passes by an accepting transition. 

\begin{theorem}
  \label{theo:AEC}
  Both, the \emph{accepted transition} and the \emph{accepted transition with update} AEC simulation, can be used to establish the good for MDPs property.
\end{theorem}
To show this, we describe the strategy translations in accordance with Corollary~\ref{cor:AECsim}.


\begin{proof}
In both cases, the translation of a winning strategy of the
spoiler for the $1$-GAEC simulation game are straightforward: The
spoiler essentially follows her winning strategy from the $1$-GAEC
simulation game, with the extra rule that she will make her AEC claim
to the duplicator on the first accepting transition on or after her
AEC claim in the $1$-GAEC claim.  If the duplicator is allowed to
update the transition, this information is ignored by the
spoiler---she plays according to her winning strategy from the
$1$-GAEC simulation game.
Naturally, the resulting play will comply with her $1$-GAEC claim, and
will thus also be winning for the---weaker---AEC claim made to the
duplicator. \qed
\end{proof}

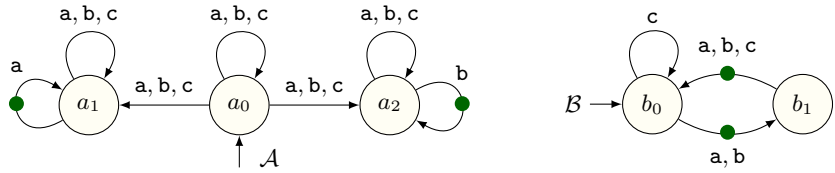
\begin{figure}[t]
  \begin{minipage}{0.59\textwidth}
    \centering
     \begin{tikzpicture}
        \node[state,initial below,fill=safecellcolor] (I) {$a_0$};
        \node[state,fill=safecellcolor] (A) [left=2cm of I] {$a_1$};
        \node[state,fill=safecellcolor] (B) [right=2cm of I] {$a_2$};
        \node (name) [below right=0.7cm and 0.4cm of I] {$\mathcal{A}$};
        \path[->]
        (I) edge [loop above] node {$\mathtt{a,b,c}$} ();
        \path[->]
        (I) edge              node[above] {$\mathtt{a,b,c}$} (A);
        \path[->]
        (I) edge              node {$\mathtt{a,b,c}$} (B);
        \path[->]
        (A) edge [loop above] node {$\mathtt{a,b,c}$} ();
        \path[->]
        (A) edge [loop left] node[accepting dot,label={[label distance=2mm]$\mathtt{a}$}] {} ();
        \path[->]
        (B) edge [loop above] node {$\mathtt{a,b,c}$} ();
        \path[->]
		(B) edge [loop right] node[accepting dot,label={[label distance=1mm]$\mathtt{b}$}] {} ();
     \end{tikzpicture}
  \end{minipage}
  \begin{minipage}{0.39\textwidth}
     \centering
     \begin{tikzpicture}
		\node[state,initial,initial text=$\mathcal{B}$,fill=safecellcolor] (I) {$b_0$};
		\node[state,fill=safecellcolor] (AB) [right=2cm of I] {$b_1$};
		\path[->, bend right]
		(I) edge 			  node[accepting dot,label={below:$\mathtt{a,b}$}] {} (AB);
		\path[->]
		(I) edge [loop above] node {$\mathtt{c}$} ();
		\path[->, bend right]
		(AB) edge 			  node [accepting dot,label={$\mathtt{a,b,c}$}] {} (I);
	 \end{tikzpicture}
  \end{minipage}
  \caption{Automata $\cal A$ (left) and $\cal B$ (right) for
    $\varphi = (\always \eventually \mathtt{a}) \vee (\always\eventually 
    \mathtt{b})$. The dotted transitions are accepting.
    The NBA $\cal A$ does not simulate the DBA $\cal B$:
    $\mathcal{B}$ can play $a$'s until $\mathcal{A}$ moves to either the state on the left, or the state on the right.
    $\mathcal{B}$ then wins by henceforth playing only $b$'s or only $a$'s.
    However, $\mathcal{A}$ is good for MDPs.
    It wins the AEC simulation game by waiting until an AEC is reached (by $\cal B$), and then
    check if $a$ or $b$ occurs infinitely often in this AEC.
    Based on this knowledge, $\cal A$ can make its decision.
    This can be shown by AEC simulation if $\cal B$ has to provide sufficient information, such as
    a list of transitions---or even a list of letters---that occur infinitely often.
    The amount of information the spoiler has to provide determines the strength of the AEC simulation used.
    If, e.g., $\cal B$ only has to reveal one accepting transition of the end-component,
    then it can select an end-component where the revealed transition
    is $(b_1,c,b_0)$, which does not provide sufficient information.
    Whereas, if the duplicator is allowed to update the transition, then the duplicator wins by updating the recorded transition
    to the next $a$ or $b$ transition}
  \label{fig:AECexample}
\end{figure}


We use AEC simulation to identify GFM automata among the automata produced (e.g., by SPOT \cite{Duret-LutzLFMRX16}) at the beginning of the transformation.
Figure \ref{fig:AECexample} shows an example for which the duplicator wins the AEC simulation game,
but loses the ordinary simulation game.
Candidates for automata to simulate are, e.g., the slim GFM B\"uchi automata and the limit deterministic B\"uchi automata discussed above.

\section{Evaluation}
\label{sec:experiments}
\subsection{General B\"uchi Automata for Probabilistic Model Checking}
\newcommand{\statstimeout}{69\xspace}
\newcommand{\statsdet}{315\xspace}
\newcommand{\statssimzero}{103\xspace}
\newcommand{\statssimone}{11\xspace}
\newcommand{\statssimtwo}{1\xspace}
\newcommand{\statsnosim}{501\xspace}
\newcommand{\statsavgspotbuchistates}{15.21\xspace}
\newcommand{\statsavgowlbuchistates}{46.35\xspace}
\newcommand{\statstotalnosimshown}{570\xspace}
\newcommand{\statsnumformulae}{1000\xspace}
\newcommand{\statstotalnumnondet}{685\xspace}

\newcommand{\ffalse}{\mathit{ff}}
\newcommand{\ttrue}{\mathit{tt}}
\newcommand{\X}{\mathsf{X}}
\newcommand{\U}{\mathbin{\mathsf{U}}}
\newcommand{\F}{\mathsf{F}}
\newcommand{\G}{\mathsf{G}}
\newcommand{\R}{\mathsf{R}}
\newcommand{\W}{\mathsf{W}}
\newcommand{\M}{\mathsf{M}}
\newcommand{\limplies}{\mathsf{\rightarrow}}
\newcommand{\liff}{\mathsf{\leftrightarrow}}
\newcommand{\lxor}{\mathsf{xor}}

As discussed, automata that simulate slim automata or SLDBAs are good for MDPs.
This fact can be used to allow B\"uchi automata produced from general-purpose tools such as SPOT's~\cite{Duret-LutzLFMRX16} ltl2tgba rather than using specialized automata types.
Automata produced by such tools are often smaller because such general-purpose tools are highly optimized and not restricted to producing slim or limit deterministic automata.
Thus, one produces an arbitrary B\"uchi automaton using any available method, then transforms this automaton into a slim or limit deterministic automaton, and finally checks whether the original automaton simulates the generated one.

We have evaluated this idea on random LTL formulas produced
by SPOT's tool randltl.
We have set the tree size, which influences the size of the formulas, to 50, and have produced \statsnumformulae formulas with 4 atomic propositions each.
We left the other values to their defaults.
We have then used SPOT's ltl2tgba (version 2.7) to turn these formulas into non-generalized B\"uchi automata using default options.
Finally, for each automaton, we have used our tool to check whether the automaton simulates a limit deterministic automaton that we produce from this automaton.
For comparison, we have also used Owl's \cite{Sicker16b} tool ltl2ldba (version 19.06.03) to compute limit deterministic non-generalized Buchi automata.
We have also used the option of this tool to compute B\"uchi automata with a nondeterministic initial part.
We used 10 minute timeouts.

Of these \statsnumformulae formulas, \statsdet can be transformed to deterministic B\"uchi automata.
For an additional \statssimzero other automata generated, standard simulation sufficed to show that they are GFM.
For a further \statssimone of them, the simplest AEC simulation (the spoiler chooses an accepting transition to occur infinitely often) sufficed, and
another \statssimtwo could be classed GFM by allowing the duplicator to update the transition.
\statsnosim automata turned out to be nonsimulatable and for \statstimeout we did not get a decision due to a timeout. 

For the LTL formulas for which ltl2tgba could not produce deterministic automata, but for which simulation could be shown,
the number of states in the generated automata was often lower than the number of states in the automata produced by Owl's tools.
On average, the number of states per automaton was $\approx$\statsavgspotbuchistates for SPOT's ltl2tgba;
while for Owl's ltl2ldba it was $\approx$\statsavgowlbuchistates.
%
We provide examples for each outcome in Table~\ref{tab:simulations}.

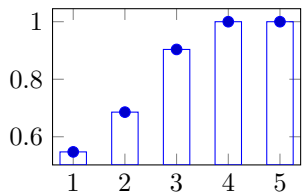
\begin{wrapfigure}[9]{l}{0.33\textwidth}
  \begin{tikzpicture}
\begin{axis}[xtick=data,width=0.4\textwidth,height=0.3\textwidth]
\addplot+[ybar,xtick={1,2,...,10}] plot coordinates
	{(1, 0.5473988842293195) (2, 0.6856633562515916) (3, 0.9035970090515545) (4, 1.0) (5, 1.0) };
\end{axis}
\end{tikzpicture}
  \vspace{-7mm}
  \caption{Deciles ratio ltl2tgba\newline/semi-deterministic automata\label{fig:rltl-decile}}
\end{wrapfigure}
\noindent Let us consider the ratio between the size of automata produced by ltl2tgba and the size of semi-deterministic automata produced by Owl.
The average of this number for all automata that are not deterministic and that can be simulated in some way is $\approx{}1.0335$.
This means that on average, for these automata, the semi-deterministic automata are slightly smaller.
If we take a look at the first $5$ deciles depicted in Fig.~\ref{fig:rltl-decile}, we see that there is a large number of formulas for which ltl2tgba and Owl produce automata of the same size.
For around $24.3478\%$ of the cases, automata by SPOT are smaller than those produced by Owl (ratio $< 1$).

\renewcommand\arraystretch{0.8}
{\footnotesize 
  \begin{table}[t]
  \centering
  \caption{Simulation results. Column ``LTL'' contains the formula the automata are generated from.
  	Column ``sim'' contains the type of simulation required to show that the original automaton is suitable for MDPs.
  	There, ``det'' means that the automaton is deterministic (no simulation required), ``sim0'' means that standard simulation sufficed, ``sim1'' means that the simulation where the spoiler is forced to choose a transition to be repeated infinitely often suffices, ``sim2'' means that in addition the witness can change the transition which has to occur infinitely often, and ``nosim'' means that we were not able to prove a simulation relation.
  	``ltl2tgba'' contains the number of states from SPOT's ltl2tgba, ``owl'' the ones from Owl's ltl2ldba
  }
  \label{tab:simulations}
  \begin{tabular}{c|c|c|c}
ltl & sim & ltl2tgba & owl \\\hline
\begin{minipage}{8cm}\begin{align*}\begin{autobreak}(\X
\F
(\mathit{ap}_{3}
\limplies
(\lnot
\mathit{ap}_{0}
\lor
\lnot
\G
\lnot
\mathit{ap}_{2}))
\land
\lnot
((\mathit{ap}_{1}
\limplies
\mathit{ap}_{3})
\R
\G
\mathit{ap}_{1}))
\lor
((\F
\mathit{ap}_{2}
\W
\mathit{ap}_{1})
\land
(\G
\mathit{ap}_{3}
\M
\mathit{ap}_{3}))\end{autobreak}\end{align*}\end{minipage} & det & 7 & 10 \\
\begin{minipage}{8cm}\begin{align*}\begin{autobreak}\G
\F
\G
(\lnot
\mathit{ap}_{3}
\lor
\X
\X
\X
(\mathit{ap}_{0}
\U
(\ffalse
\R
\mathit{ap}_{2})))
\R
\F
\G
\mathit{ap}_{1}\end{autobreak}\end{align*}\end{minipage} & sim0 & 2 & 2 \\
\begin{minipage}{8cm}\begin{align*}\begin{autobreak}\lnot
\G
\F
\lnot
(\X
(\mathit{ap}_{1}
\liff
(\ttrue
\U
\X
\mathit{ap}_{2}))
\lxor
\X
(\ffalse
\R
(\lnot
\mathit{ap}_{0}
\liff
\G
\X
\mathit{ap}_{3})))\end{autobreak}\end{align*}\end{minipage} & sim1 & 26 & 14 \\
\begin{minipage}{8cm}\begin{align*}\begin{autobreak}\F
\G
(\lnot
(\F
\mathit{ap}_{2}
\R
(\mathit{ap}_{1}
\land
\G
\mathit{ap}_{0}))
\limplies
(\mathit{ap}_{3}
\lxor
\G
\mathit{ap}_{1}))
\liff
(\G
\mathit{ap}_{2}
\land
(\F
\X
\mathit{ap}_{0}
\U
\lnot
(\mathit{ap}_{2}
\land
\lnot
\mathit{ap}_{2})))\end{autobreak}\end{align*}\end{minipage} & sim2 & 24 & 41 \\
\begin{minipage}{8cm}\begin{align*}\begin{autobreak}\lnot
(((\mathit{ap}_{3}
\limplies
(\mathit{ap}_{2}
\liff
\mathit{ap}_{3}))
\M
\mathit{ap}_{1})
\R
\mathit{ap}_{3})
\W
\lnot
\F
\G
\mathit{ap}_{1}\end{autobreak}\end{align*}\end{minipage} & nosim & 5 & 10 \\
\end{tabular}

\end{table}
}


\subsection{GFM Automata and Reinforcement Learning}
\label{sec:gfm-rl}

SLDBAs have been used in \cite{Hahn19} for model-free reinforcement
learning of $\omega$-regular objectives.  While the B\"uchi acceptance
condition allows for a faithful translation of the objective to a scalar
reward, the agent has to learn how to control the automaton's
nondeterministic choices; that is, the agent has to learn when the
SLDBA should cross from the initial component to the accepting
component to produce a successful run of a behavior that satisfies
the given objective.

Any GFM automaton with a B\"uchi acceptance condition can be used
instead of an SLDBA in the approach of \cite{Hahn19}.  While in many
cases SLDBAs work well,
(see, for example, results for randomly generated problems in Table~\ref{tab:simulations})
GFM automata that are not limit-deterministic may provide a
significant advantage.

Early during training, the agent relies on uniform random choices to
discover policies that lead to successful episodes.  This includes
randomly resolving the automaton nondeterminism.  If random choices
are unlikely to produce successful runs of the automaton in case of
behaviors that should be accepted, learning is hampered because good
behaviors are not rewarded.  Therefore, GFM automata that are more
likely to accept under random choices will result in the agent
learning more quickly.
We have found the following properties of GFM automata to affect
the agent's learning ability.

\noindent\textbf{Low branching degree.}
A low branching degree presents the agent with fewer alternatives,
reducing the expected number of trials before the agent finds a good
combination of choices.  Consider an MDP and an automaton that
require a specific sequence of $k$ nondeterministic choices
in order for the automaton to accept.  If at each choice there
are $b$ equiprobable options, the correct sequence is obtained with
probability $b^{-k}$.


This is a simplified description of what happens in the \texttt{milk}
example, described in Appendix~\ref{sec:rl-exp}.  In this example, the
agent learns in fewer episodes with a slim automaton than with other
GFM automata with higher branching degrees.

\noindent\textbf{Cautiousness.}  An automaton that enables fewer
nondeterministic choices for the same finite input word gives the
agent fewer chances to choose wrong.  The slim automata construction
has the interesting property of ``collecting hints of acceptance''
before a nondeterministic choice is enabled because $S'$ has to be
nonempty for a $\gamma_{2,1}$ transition to be present and that
requires going through at least one accepting transition.

Consider a model with the objective $\eventually\always p$, in which
$p$ changes value at each transition many times before a winning end
component may be reached.  An SLDBA for $\eventually\always p$ that
may move to the accepting part whenever it reads $p$ is very likely to
jump prematurely and reject.  However, an automaton produced by the
breakpoint construction must see $p$ twice in a row before the jump is
enabled, preventing the agent from making mistakes in its control.
(See example \texttt{oddChocolates} in Appendix~\ref{sec:rl-exp}.)

\noindent\textbf{Forgiveness.} Mistakes made in resolving nondeterminism may
be irrecoverable.  This is often true of SLDBAs meant for
model checking, in which jumps are made to select a subformula to be
eventually satisfied.  However, general GFM automata, thanks
also to their less constrained structure, may be constructed to
``forgive mistakes'' by giving more chances of picking a successful
run.

Figure~\ref{fig:forgiveness} compares a typical SLDBA to an automaton
that is not limit-deterministic and is not produced by the breakpoint
construction, but is proved GFM by AEC simulation.  This latter
automaton has a nondeterministic choice in state $q_0$ on letter
$x \wedge \neg y$ that can be made an unbounded number of times.  The
agent may choose $q_1$ repeatedly even if eventually
$\eventually\always x$ is false and $\always\eventually y$ is true.
With the SLDBA, on the other hand, there is no room for error.
A model where the added flexibility improves learning is
\texttt{forgiveness} in Appendix~\ref{sec:rl-exp}.

Note that for the automata from Figure \ref{fig:forgiveness}, the state $q_0$
from the forgiving automaton on the right simulates all three states
of the SLDBA on the left, $q_1$ simulates $a_1$ and $a_2$, and $q_2$
simulates $a_2$, but no state of the the SLDBA simulates any state
of the forgiving automaton.


\begin{figure}
  \centering
  \begin{tikzpicture}[
    every text node part/.style={align=center},
    every state/.style={fill=safecellcolor}]
    \begin{scope}
      \node[state,initial above] (Q0) at (0,0) {$a_0$};
      \node[state] (Q1) [right=1.75cm of Q0] {$a_1$};
      \node[state] (Q2) [left=1.75cm of Q0] {$a_2$};
      \path[->]
      (Q0) edge [loop below] node {$\top$} ()
      edge node [accepting dot,label={$x$}] {} (Q1)
      edge node [accepting dot,label={$y$}] {} (Q2)
      (Q1) edge [loop above] node [accepting dot,label={$x$}] {} ()
      (Q2) edge [loop above] node [accepting dot,label={$y$}] {} ()
      (Q2) edge [loop below] node {$\neg y$} ();
    \end{scope}
    \begin{scope}[xshift=3.75cm]
      \node[state,initial] (Q0) at (0,0) {$q_0$};
      \node[state] (Q1) [right=2.25 cm of Q0] {$q_1$};
      \node[state] (Q2) [right=2.25 cm of Q1] {$q_2$};
      \path[->]
      (Q0) edge [loop above] node [accepting dot,label={$y$}] {} ()
      edge [loop below] node {$\neg y$} ()
      edge [bend left] node [accepting dot,label={$x \wedge \neg y$}] {} (Q1)
      (Q1) edge [loop above] node [accepting dot,label={$x \wedge \neg
        y$}] {} ()
      edge [bend left] node [accepting dot, label={$y$}] {} (Q0)
      edge node {$\neg x \wedge \neg y$} (Q2)
      (Q2) edge [loop right] node {$\neg y$} ()
      edge [bend left=40] node [accepting dot, label={below:$y$}] {} (Q0);
    \end{scope}
  \end{tikzpicture}
  \caption{Two GFM automata for
    $(\eventually\always x) \vee (\always\eventually y)$.  SLDBA
    (left), and forgiving (right)}
  \label{fig:forgiveness}
\end{figure}
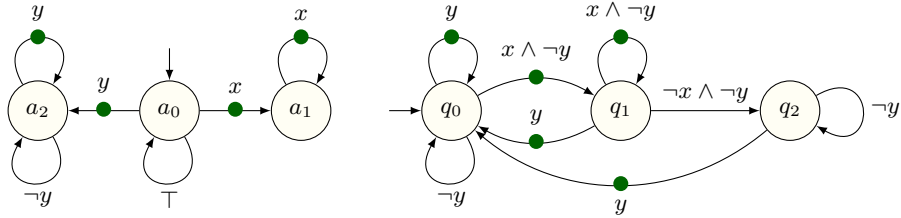


\noindent\textbf{A Case Study.}
We compared the effectiveness in learning to control a cart-pole
model of three automata for the property
$\bigl((\eventually\always x) \vee (\always\eventually y)\bigr) \wedge
\always \mathtt{safe}$.
The safety component of the objective is to keep the pole balanced and
the cart on the track.  The left two thirds of the track alternate
between $x$ and $y$ at each step.  The right third is always labeled
$y$, but in order to reach it, the cart has to cross a barrier, with
probability $1/3$ of failing.

The three automata are an SLDBA (4 states), a slim automaton (8
states), and a handcrafted forgiving automaton (4 states) similar to
the one of Fig.~\ref{fig:forgiveness}.

\begin{figure}[t]
  \centering
  \includegraphics[scale=0.65,trim=0cm 0.2cm 0cm 1cm,clip]{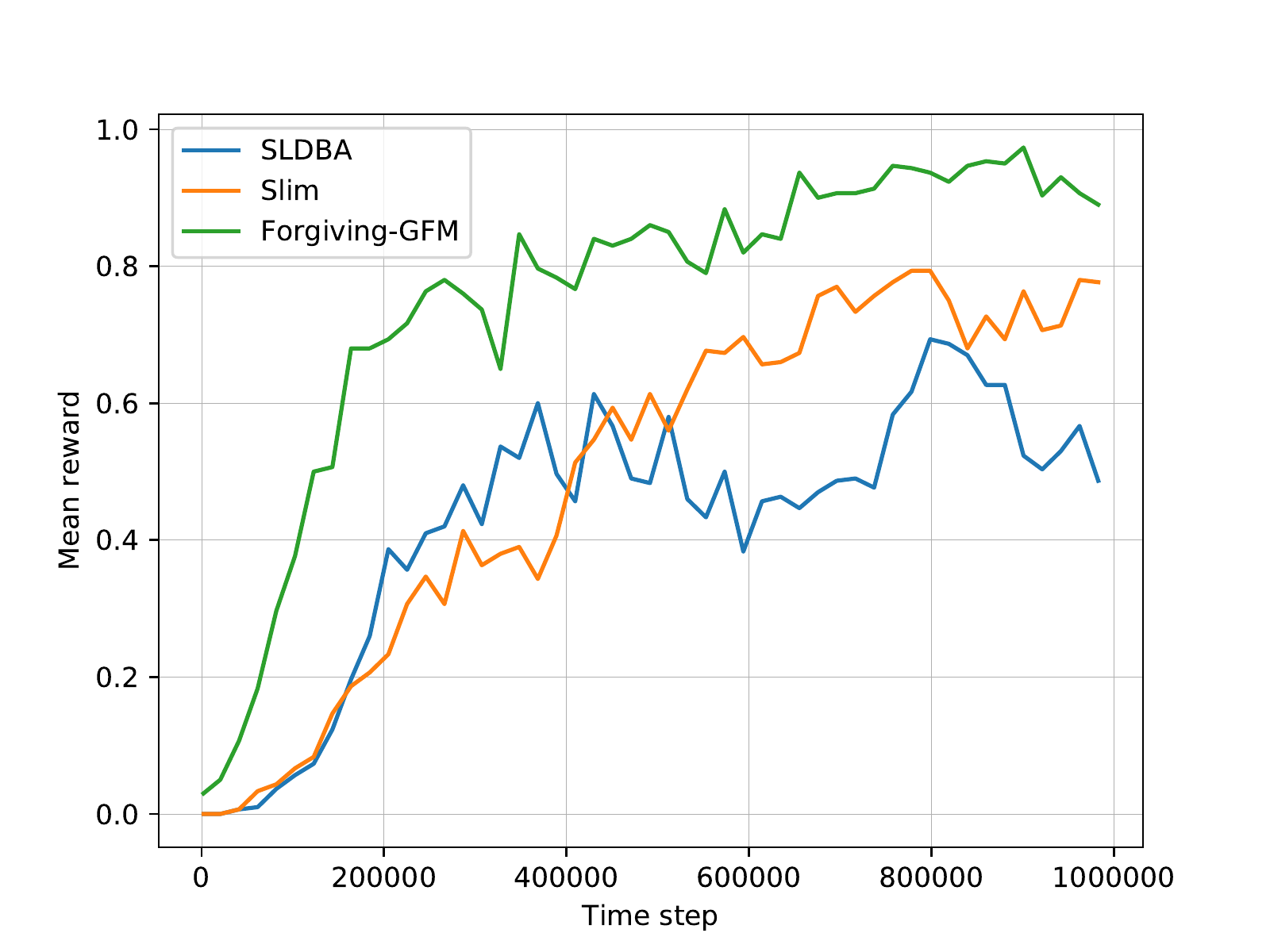}
  \caption{Learning curves}
  \label{fig:learn}
\end{figure}
Training of the continuous-statespace model employed Proximal Policy
Optimization (PPO) \cite{Schulm17}
as implemented in OpenAI Baselines \cite{baselines}. Most training 
parameters were left unmodified from their defaults. This includes using 
the Adam optimizer \cite{kingma2014adam} with a learning rate of 
$3\times10^{-4}$, a discount factor $\gamma = 0.99$, and $K = 4$ epochs 
of optimization per iteration. Each iteration consisted of 2048 steps 
sampled from 12 parallel environments. The clipping parameter was changed
to $\epsilon = 0.1$. We also modified the default network architecture 
to contain  two hidden layers with 128 neurons per layer. Each hidden 
layer was followed by a hyperbolic tangent activation. The training 
parameters remained the same for each automaton used in our experiment.

%
%
%
Figure~\ref{fig:learn} shows the learning curves for the three
automata averaged over three runs.  They underline the importance of
choosing the right automaton in RL.




\section{Conclusion}
\label{sec:conclusion}
We have defined the class of automata that are \emph{good for MDPs}---nondeterministic automata that can be used for the analysis of MDPs---and shown it to be closed under different simulation relations.
This has multiple favorable implications for model checking and reinforcement learning.
Closure under classic simulation opens a rich toolbox of statespace reduction techniques that come in handy to push the boundary of analysis techniques,
while the more powerful (and more expensive) AEC simulation has
promise to identify source automata that happen to be good for MDPs.
Using standard and (very basic) AEC simulation, we only failed to establish the good for MDPs property for the NBAs created for 
$\statstotalnosimshown$ of $\statstotalnumnondet$ random LTL examples that are not expressible as DBAs, $\statstimeout$ of them due to timeout.

The wider class of GFM automata also shows promise:
the slim automata we have defined to tame the branching degree while
retaining the desirable B\"uchi condition for reinforcement learning
are
able to compete even against optimized SLDBAs.

As outlined in Section \ref{sec:gfm-rl}, a low branching
degree, cautiousness, and forgiveness  make
automata particularly well-suited for learning.
A low branching degree is the easiest among them to quantify, and the \emph{slim automata} tick this box perfectly.
We have also argued that they are contributing towards cautiousness.
For forgiveness, on the other hand, we had to hand-craft automata that satisfy this property.
From a practical point of view, much of the power of this new approach
is in harnessing the power of simulation for learning,
and forgiveness is closely related to simulation.

The natural follow-up research is to tap the full
potential of simulation-based statespace reduction instead of the
limited version---which preserves limit determinism---
that we have implemented.
Besides using this to get the statespace small---useful for model checking---we will
use simulation to construct forgiving automata, which is promising for reinforcement learning.


\clearpage
\bibliographystyle{plain}
\bibliography{papers}

\clearpage

\appendix

\section{Reinforcement Learning Case Studies}
\label{sec:rl-exp}

We collect here the models and properties discussed in
Section~\ref{sec:gfm-rl}.  The \texttt{milk} model, shown below as
PRISM code \cite{Kwiatk11}, shows the
effects of low branching degree.  The property of interest is
$\bigwedge_{0 \leq i \leq 4} \eventually\always p_i$.  The slim
automaton has 275 states, while the SLDBA automaton has 7 states.
Since the $p_i$'s are not mutually exclusive, the slim automaton needs
a large number of states to limit the branching degree.  Yet, for a
fixed set of hyperparameter values, the agent learns the optimal
policy with the slim automaton and fails to learn it with the SLDBA.

\begin{lstlisting}[language=prism]
mdp
const int M = 12;

module m
  b : [0..M] init M;
  [a] true -> (b' = b > 0 ? b-1 : b);
endmodule

label "p0" = b > 0;
label "p1" = b > 1;
label "p2" = b > 2;
label "p3" = b > 3;
label "p4" = b = 0;
\end{lstlisting}
The \texttt{oddChocolates} model shows the benefits of cautiousness in
resolving nondeterminism.  The property is
$\eventually\always \mathtt{odd}$.  The slim automaton has 4 states
and is limit-deterministic, while the usual SLDBA, optimized for size,
has 3 states.  The extra state produced by the breakpoint construction
allows the slim automaton to jump to the accepting part only after
seeing two consecutive ``odd'' states.  As a result, for a fixed set
of hyperparameter values, and the tool \textsc{Mungojerrie}
\cite{Hahn19}, the agent reliably learns an optimal policy, unlike the
case of the smaller SLDBA automaton.
\begin{lstlisting}[language=prism]
mdp
const int N = 5;
const int M = 12;

module pluck
  b0 : [0..M] init M;
  b1 : [0..M] init M;
  b2 : [0..M] init M;
  b3 : [0..M] init M;
  b4 : [0..M] init M;
  [a0] true -> (b0' = b0 > 0 ? b0-1 : b0);
  [a1] true -> (b1' = b1 > 0 ? b1-1 : b1);
  [a2] true -> (b2' = b2 > 0 ? b2-1 : b2);
  [a3] true -> (b3' = b3 > 0 ? b3-1 : b3);
  [a4] true -> (b4' = b4 > 0 ? b4-1 : b4);
endmodule

label "odd" = mod(b0+b1+b2+b3+b4,2) = 1;
\end{lstlisting}
The \texttt{forgiveness} model below illustrates the effects of
forgiveness.  Its main feature is to produce, while in its transient
states, deceiving hints for the learner.  Two automata for
$(\eventually\always x) \vee (\always\eventually y)$ are shown in
Figure~\ref{fig:forgiveness}.  The slim automaton produced by the
breakpoint construction is not limit-deterministic.  It accepts
behaviors that satisfy $\always\eventually y$ before (irreversibly)
jumping and behaviors that satisfy $\eventually\always x$ after
jumping.  Therefore the agent can recover from mistakenly believing
that an MDP run will satisfy $\always\eventually y$, while it ends up
satisfying $\eventually\always x$, but not vice versa.  Accordingly,
the results with the slim automaton are in between those of the two
automata of Figure~\ref{fig:forgiveness}.
\begin{lstlisting}[language=prism]
mdp
const double p = 1/2;
const double q = 0.1;
const int N = 5;

module m
  s : [0..7] init 0;
  d : [0..N] init N;
  [a] s=0 -> p : (s'=1) + (1-p) : (s'=2);
  [b] s=0 -> p : (s'=2) + (1-p) : (s'=3);
  [c] (s=1 | s=2) & d > 0 -> q : (d'=d-1) + (1-q) : true;
  [c] (s=1 | s=2) & d = 0 -> (s'=s+3);
  [c] s=3 -> true;
  [a] s=4 | s=5 -> q : (s'=s+2) + (1-q) : true;
  [b] s=4 | s=5 -> true;
  [c] s=6 | s=7 -> (s'=s-2);
endmodule

label "x" = s=1 | s=5;
label "y" = s=2 | s=6;
\end{lstlisting}


\section{Proving Simulation}
\label{sec:proving-simulation}
\subsection{Classic Simulation}
In the following, we discuss how we can prove that a given B\"uchi automaton ${\cal D} = \langle \Sigma,Q_{\cal D},Q_{{\cal D},0},\Delta_{\cal D},\Gamma_{\cal D} \rangle$ (duplicator) simulates a B\"uchi automaton ${\cal S} = \langle \Sigma,Q_{\cal S},Q_{{\cal S},0},\Delta_{\cal S},\Gamma_{\cal S} \rangle$ (spoiler).
The basic idea is to construct a parity game $\cal G$ such that the even player wins if and only if the simulation holds.
The construction of the game is fairly standard and closely follows the outline from Subsection ~\ref{ssec:simulate}.

In our notion, we define parity games as follows.
\begin{definition}
  A \emph{transition-labelled maximum parity game} is a tuple
  \[
    {\cal P} = \langle Q_0, Q_1, \Delta, c \rangle
  \]
  where for $Q = Q_0 \cup Q_1$ we have
  \begin{itemize}
  \item $Q_0 \cap Q_1 = \emptyset$,
  \item $\Delta \subseteq Q \times Q$,
  \item $c\colon Q \times Q \rightharpoonup \mathbb{N}$ where $\mathit{Dom}(c) = \Delta$.
  \end{itemize}

  An \emph{infinite play} is an infinite sequence $p = q_0 q_1 \ldots \in Q^\omega$ such for for all $i \geq 0$ it is $(q_i, q_{i+1}) \in \Delta$.
  An infinite play $p$ is \emph{winning} if the highest colour which appears infinitely often is even, that is
  $\max \left\{ c' \mid \forall i \geq 0 \exists j \geq i.\ c' = c(q_j, q_{j+1})\right\}\ \mathrm{mod}\ 2 = 0$.
  By $\mathrm{infplays}({\cal P})$ we denote the set of all infinite plays of ${\cal P}$.
  A \emph{finite play} is a finite sequence $p = q_0 q_1 \ldots q_n \in Q^*$ such for for all $i, 0 \leq i \leq n - 1$ it is $(q_i, q_{i+1}) \in \Delta$.
  By $\mathrm{finplays}_i({\cal P})$ we denote the set of all finite plays $p$ of ${\cal P}$ such that $q_n \in Q_i$ for $i \in \{0, 1\}$.
  A \emph{strategy} for player $i$ is a function $f_i\colon \mathrm{finplays}_i({\cal P}) \to Q$ such that for $p = q_0 q_1 \ldots q_n$ we have $(q_n, f(p)) \in \Delta$.
  Given a pair of a player $0$ and player $1$ strategies $\langle f_0, f_1\rangle$, the \emph{induced play} starting from state $q$ is
  $p = q_0 q_1 q_2 \ldots$
  where $q_0 = q$ and for all $k \geq 0$ we have $q_{k+1} = f_i(q_k)$ where $q_k \in \Delta_i$.
  A player $0$ strategy $f_0$ is \emph{winning} in a state $q$ if for every player $1$ strategy $f_1$ the induced play is winning.
\end{definition}
Note that here colours are attached to transitions.
This change is only for technical convenience, because colours naturally occur on the edges rather than states of the games we construct.
Models of this type can either be transformed to state-labelled parity games or solved by slightly adapted versions of existing algorithms.

The parity game we construct is then as follows:
\begin{definition}
  Consider two B\"uchi automata ${\cal S} = \langle \Sigma,Q_{\cal S},Q_{{\cal S},0},\Delta_{\cal S},\Gamma_{\cal S} \rangle$ and ${\cal D} = \langle \Sigma,Q_{\cal D},Q_{{\cal D},0},\Delta_{\cal D},\Gamma_{\cal D} \rangle$.
  The \emph{classical simulation parity game} ${\cal P}_{{\cal S}, {\cal D},\mathrm{sim0}}$ is defined as
  \[
    {\cal P}_{{\cal S}, {\cal D},\mathrm{sim0}} = \langle Q_0, Q_1, \Delta, c \rangle
    \]
    where
    \begin{itemize}
    \item $Q_0 = Q_{\cal S} \times Q_{\cal D} \times \Sigma$,
    \item $Q_1 = Q_{\cal S} \times Q_{\cal D}$,
    \item $\Delta = \Delta_1 \cup \Delta_0$ where
    \item $\Delta_0 = \left\{ ((q_{\cal S}, q_{\cal D}, a), (q_{\cal S}, q_{\cal D}')) \mid (q_{\cal D}, a, q_{\cal D}') \in \Delta_{\cal D} \right\}$
    \item $c((q_{\cal S}, q_{\cal D}, a), (q_{\cal S}, q_{\cal D}')) =
    \begin{cases}
    2 & (q_{\cal D}, a, q_{\cal D}') \in \Gamma_{\cal D}\\
    0 & \text{else}
    \end{cases}$
    \item $\Delta_1 = \left\{ ((q_{\cal S}, q_{\cal D}), (q_{\cal S}', q_{\cal D}, a)) \mid (q_{\cal S}, a, q_{\cal S}') \in \Delta_{\cal S} \right\}$
    \item $c((q_{\cal S}, q_{\cal D}), (q_{\cal S}', q_{\cal D}, a)) =
    \begin{cases}
    1 & (q_{\cal S}, a, q_{\cal S}') \in \Gamma_{\cal S}\\
    0 & \text{else}
    \end{cases}$
    \end{itemize}
\end{definition}
The states $Q_1 = Q_{\cal S} \times Q_{\cal D}$ are the states controlled by the spoiler.
A state $(q_{\cal S}, q_{\cal D})$ represents the situation where the spoiler is in state $q_{\cal S}$ and the duplicator is in state $q_{\cal D}$.
States $Q_0 = Q_{\cal S} \times Q_{\cal D} \times \Sigma$ are controlled by the duplicator.
A state $(q_{\cal S}, q_{\cal D}, a)$ represents the situation where the spoiler has chosen a transition with label $a$ and has moved to state $q_{\cal S}$ then;
the component $a$ has to be included in the state description, because the duplicator has to react with a transition with the same label as the one chosen by the spoiler.
The transitions of $\Delta_1$ are the ones of the spoiler.
The spoiler can choose any input $a$ from its current state $q_{\cal S}$ and then move to its successor state $q_{\cal S}'$ if $(q_{\cal S}, a, q_{\cal S}')$ is a valid transition in the spoiler automaton.
If this transition is an accepting transition (in the spoiler automaton), then it will have colour $1$, otherwise $0$.
The transitions of $\Delta_0$ are the ones of the duplicator.
In a state $(q_{\cal S}, q_{\cal D}, a)$, the duplicator has to choose a transition (in the B\"uchi automaton) labelled with $a$, because this was the choice of input of the spoiler.
If this transition is accepting, the transition in the game will have label $2$, otherwise it will have label $0$.

Now, a play of the constructed parity game exactly describes a simulation game of Subsection~\ref{ssec:simulate}.
Also, the play is winning if and only if the simulation game is winning:
If the duplicator wins the game, then either a) neither the spoiler nor the duplicator have infinitely many accepting transitions or b) the duplicator has infinitely many accepting transitions.
In case a), the resulting colour of the play will be $0$.
In case b), the resulting colour will be $2$.
In both cases, it is even and the even player wins.
Because of this, we have the following:

\begin{lemma}
  \label{lem:classicalsimdef}
  Consider B\"uchi automata ${\cal S} = \langle \Sigma,Q_{\cal S},Q_{{\cal S},0},\Delta_{\cal S},\Gamma_{\cal S} \rangle$ and ${\cal D} = \langle \Sigma,Q_{\cal D},Q_{{\cal D},0},\Delta_{\cal D},\Gamma_{\cal D} \rangle$ and the classical simulation parity game ${\cal P}_{{\cal S}, {\cal D}}$.
  Assume that for all $q_{\cal S} \in Q_{{\cal S},0}$ we have $q_{\cal D} \in Q_{{\cal D},0}$ such that $(q_{\cal S}, q_{\cal D})$ is a winning state in the parity game.
  Then ${\cal D}$ simulates ${\cal S}$.
  Otherwise, ${\cal D}$ does not simulate ${\cal S}$.
\end{lemma}

To use the above lemma to show that simulation holds, the usual algorithms for solving parity games can be used (with a slight adaption to take into account the fact that the edges are labelled with a parity rather than the states).
We have implemented a variation of the McNaughton algorithm~\cite{McNaug66}.
Note that the parity games we construct have no more than three colours, for which a specialised algorithms exists~\cite{EtessamiWS05}.
However, because in our experience it always took much more time to construct the parity game than to solve this, we did not apply this specialised algorithm so far.

\subsection{AEC Simulation}
Classic simulation is not required to preserve good for MDP-ness, and the above parity game might fail to demonstrate a good for MDP property.
In Subsection~\ref{ssec:aec.simulate}, we have discussed of how the simulation can be refined such that more automata can be proven to be sufficiently similar for our purpose (AEC simulation).
In the next definition, we force the spoiler to choose a transition which is to be repeated infinitely often.

\begin{definition}
  \label{def:classicalsimdef}
  Consider two B\"uchi automata ${\cal S} = \langle \Sigma,Q_{\cal S},Q_{{\cal S},0},\Delta_{\cal S},\Gamma_{\cal S} \rangle$ and ${\cal D} = \langle \Sigma,Q_{\cal D},Q_{{\cal D},0},\Delta_{\cal D},\Gamma_{\cal D} \rangle$.
  The \emph{accepting transition simulation parity game} ${\cal P}_{{\cal S}, {\cal D},\mathrm{sim1}}$ is defined as
  \[
    {\cal P}_{{\cal S}, {\cal D},\mathrm{sim1}} = \langle Q_0, Q_1, \Delta, c \rangle
    \]
    where
    \begin{itemize}
    \item $Q_0 = Q_{\cal S} \times Q_{\cal D} \times \Sigma \cup Q_{\cal S} \times Q_{\cal D} \times \Sigma \times (Q_{\cal S} \times \Sigma \times Q_{\cal S})$,
    \item $Q_1 = Q_{\cal S} \times Q_{\cal D} \cup Q_{\cal S} \times Q_{\cal D} \times (Q_{\cal S} \times \Sigma \times Q_{\cal S})$,
    \item $\Delta = \Delta_{0,i} \cup \Delta_{1,i} \cup \Delta_{1,t} \cup \Delta_{0,f} \cup \Delta_{1,f} $ where
    \item $\Delta_{0,i} = \left\{ ((q_{\cal S}, q_{\cal D}, a), (q_{\cal S}, q_{\cal D}')) \mid (q_{\cal D}, a, q_{\cal D}') \in \Delta_{\cal D} \right\}$
    \item $c((q_{\cal S}, q_{\cal D}, a), (q_{\cal S}, q_{\cal D}')) = 0$
    \item $\Delta_{1,i} = \left\{ ((q_{\cal S}, q_{\cal D}), (q_{\cal S}', q_{\cal D}, a)) \mid (q_{\cal S}, a, q_{\cal S}') \in \Delta_{\cal S}\right\}$,
    \item $c((q_{\cal S}, q_{\cal D}), (q_{\cal S}', q_{\cal D}, a)) = 0$,
    \item $\Delta_t = \left\{ ((q_{\cal S}, q_{\cal D}), (q_{\cal S}', q_{\cal D}, a, (q_{\cal S}, a, q_{\cal S}'))) \mid (q_{\cal S}, a, q_{\cal S}') \in \Gamma_{\cal S}\right\}$
    \item $c((q_{\cal S}, q_{\cal D}), (q_{\cal S}', q_{\cal D}, a, (q_{\cal S}, a, q_{\cal S}'))) = 0$,
    \item $\Delta_{0,f} = \left\{ ((q_{\cal S}, q_{\cal D}, a, e), (q_{\cal S}, q_{\cal D}', e)) \mid (q_{\cal D}, a, q_{\cal D}') \in \Delta_{\cal D} \right\}$
    \item $c((q_{\cal S}, q_{\cal D}, a, e), (q_{\cal S}, q_{\cal D}', e)) = \begin{cases}
    2 & (q_{\cal D}, a, q_{\cal D}') \in \Gamma_{\cal D}\\
    0 & \text{else}
    \end{cases}$
    \item $\Delta_{1,f} = \left\{ ((q_{\cal S}, q_{\cal D}, e), (q_{\cal S}', q_{\cal D}, a, e)) \mid (q_{\cal S}, a, q_{\cal S}') \in \Delta_{\cal S} \right\}$
    \item $c((q_{\cal S}, q_{\cal D}, e), (q_{\cal S}', q_{\cal D}, a, e)) = \begin{cases}
    1 & (q_{\cal S}, a, q_{\cal S}') = e\\
    0 & \text{else}
    \end{cases}$
    \end{itemize}
\end{definition}
For the two players, there are now two different types of states.
The first type is like the state for the classical simulation parity game.
The second type is extended by $(Q_{\cal S} \times \Sigma \times Q_{\cal S})$, such that it can store edges of the spoiler.
There are now three types of transitions:
$\Delta_{0,i}$ and $\Delta_{1,i}$ are similar to the ones of $\Delta_0$ and $\Delta_0$ of the classical simulation parity game.
The difference is that the colour is always $0$.
This implies that the spoiler cannot win by just using this type of transitions.
Instead, she has to choose a transition of $\Delta_t$, which corresponds to choosing an edge to occur infinitely often.
To reduce the size of the state-space and to ease definition, we restrict this choice such that it can only choose accepting edges, because it cannot win with non-accepting edges anyway.
For the same reason, we only allow choosing the edge it just took.
Transitions $\Delta_{0,f}$ and $\Delta_{1,f}$ correspond to the situation where the spoiler has already chosen an edge.
Transitions of $\Delta_{0,f}$ are labelled with $2$ if they correspond to an accepting transition of the duplicator and $0$ otherwise.
Transitions of $\Delta_{1,f}$ are labelled with $1$ if the edge is taken which the spoiler chose to appear infinitely often.

\begin{lemma}
  \label{lem:accsimdef}
  Consider B\"uchi automata ${\cal S} = \langle \Sigma,Q_{\cal S},Q_{{\cal S},0},\Delta_{\cal S},\Gamma_{\cal S} \rangle$ and ${\cal D} = \langle \Sigma,Q_{\cal D},Q_{{\cal D},0},\Delta_{\cal D},\Gamma_{\cal D} \rangle$ and the accepting simulation parity game ${\cal P}_{{\cal S}, {\cal D}}$.
  Assume that for all $q_{\cal S} \in Q_{{\cal S},0}$ we have $q_{\cal D} \in Q_{{\cal D},0}$ such that $(q_{\cal S}, q_{\cal D})$ is a winning state in the parity game.
  Then ${\cal D}$ accepting transition simulates ${\cal S}$.
  Otherwise, ${\cal D}$ does not simulate ${\cal S}$.
\end{lemma}

Lemma~\ref{lem:accsimdef} can be used to prove or disprove simulation in the same way as Lemma~\ref{lem:classicalsimdef}.

The simulation relation in which the duplicator is allowed to change the edge to be repeated infinitely often can be implemented in the following parity game:

\begin{definition}
  Consider two B\"uchi automata ${\cal S} = \langle \Sigma,Q_{\cal S},Q_{{\cal S},0},\Delta_{\cal S},\Gamma_{\cal S} \rangle$ and ${\cal D} = \langle \Sigma,Q_{\cal D},Q_{{\cal D},0},\Delta_{\cal D},\Gamma_{\cal D} \rangle$.
  The \emph{accepting transition with update simulation parity game} ${\cal P}_{{\cal S}, {\cal D},\mathrm{sim2}}$ is defined as
  \[
    {\cal P}_{{\cal S}, {\cal D},\mathrm{sim2}} = \langle Q_0, Q_1, \Delta, c \rangle
    \]
    where
    \begin{itemize}
    \item $Q_0 = Q_{\cal S} \times Q_{\cal D} \times \Sigma \cup Q_{\cal S} \times Q_{\cal D} \times (Q_{\cal S} \times \Sigma \times Q_{\cal S}) \times (Q_{\cal S} \times \Sigma \times Q_{\cal S})$,
    \item $Q_1 = Q_{\cal S} \times Q_{\cal D} \cup Q_{\cal S} \times Q_{\cal D} \times (Q_{\cal S} \times \Sigma \times Q_{\cal S})$,
    \item $\Delta = \Delta_{0,i} \cup \Delta_{1,i} \cup \Delta_{1,t} \cup \Delta_{0,f} \cup \Delta_e \cup \Delta_{1,f} $ where
    \item $\Delta_{0,i} = \left\{ ((q_{\cal S}, q_{\cal D}, a), (q_{\cal S}, q_{\cal D}')) \mid (q_{\cal D}, a, q_{\cal D}') \in \Delta_{\cal D} \right\}$
    \item $c((q_{\cal S}, q_{\cal D}, a), (q_{\cal S}, q_{\cal D}')) = 0$,
    \item $\Delta_{1,i} = \left\{ ((q_{\cal S}, q_{\cal D}), (q_{\cal S}', q_{\cal D}, a)) \mid (q_{\cal S}, a, q_{\cal S}') \in \Delta_{\cal S}\right\}$,
    \item $c((q_{\cal S}, q_{\cal D}), (q_{\cal S}', q_{\cal D}, a)) = 0$
    \item $\Delta_t = \left\{ ((q_{\cal S}, q_{\cal D}), (q_{\cal S}', q_{\cal D}, (q_{\cal S}, a, q_{\cal S}'), (q_{\cal S}, a, q_{\cal S}'))) \mid (q_{\cal S}, a, q_{\cal S}') \in \Gamma_{\cal S}\right\}$
    \item $c((q_{\cal S}, q_{\cal D}), 0, (q_{\cal S}', q_{\cal D}, (q_{\cal S}, a, q_{\cal S}'), (q_{\cal S}, a, q_{\cal S}'))) = 0$,
    \item $\Delta_{0,f} = \left\{ ((q_{\cal S}, q_{\cal D}, (\hat{q_{\cal S}}, a, \hat{q_{\cal S}}), e), (q_{\cal S}, q_{\cal D}', e)) \mid (q_{\cal D}, a, q_{\cal D}') \in \Delta_{\cal D} \right\}$,
    \item $c((q_{\cal S}, q_{\cal D}, (\hat{q_{\cal S}}, a, \hat{q_{\cal S}}), e), (q_{\cal S}, q_{\cal D}', e)) = \begin{cases}
    2 & (q_{\cal D}, a, q_{\cal D}') \in \Gamma_{\cal D}\\
    0 & \text{else}
    \end{cases}$
    \item $\Delta_e = \left\{ ((q_{\cal S}, q_{\cal D}, (\hat{q_{\cal S}}, a, \hat{q_{\cal S}}), e), (q_{\cal S}, q_{\cal D}', (\hat{q_{\cal S}}, a, \hat{q_{\cal S}}))) \mid (q_{\cal D}, a, q_{\cal D}') \in \Delta_{\cal D} \right\}$,
    \item $c((q_{\cal S}, q_{\cal D}, (\hat{q_{\cal S}}, a, \hat{q_{\cal S}}), e), (q_{\cal S}, q_{\cal D}', (\hat{q_{\cal S}}, a, \hat{q_{\cal S}}))) = \begin{cases}
    2 & (q_{\cal D}, a, q_{\cal D}') \in \Gamma_{\cal D}\\
    0 & \text{else}\\
    \end{cases}$
    \item $\Delta_{1,f} = \left\{ ((q_{\cal S}, q_{\cal D}, e), (q_{\cal S}', q_{\cal D}, (q_{\cal S}, a, q_{\cal S}'), e)) \mid (q_{\cal S}, a, q_{\cal S}') \in \Delta_{\cal S} \right\}$,
    \item $c((q_{\cal S}, q_{\cal D}, e), (q_{\cal S}', q_{\cal D}, (q_{\cal S}, a, q_{\cal S}'), e)) = \begin{cases}
    1 & (q_{\cal S}, a, q_{\cal S}') = e\\
    0 & \text{else}
    \end{cases}$
    \end{itemize}
\end{definition}
Compared to the accepting transition simulation parity game, we have changed the second type of duplicator states.
The spoiler usually chooses an input, and then proceeds to a successor state with this input.
Rather than just storing the input, we store the whole edge which the spoiler chose.
This allows to change the edge which is expected to be seen infinitely often for the spoiler to win by using the transitions of type $\Delta_e$.

\begin{lemma}
  \label{lem:accupdatesimdef}
  Consider B\"uchi automata ${\cal S} = \langle \Sigma,Q_{\cal S},Q_{{\cal S},0},\Delta_{\cal S},\Gamma_{\cal S} \rangle$ and ${\cal D} = \langle \Sigma,Q_{\cal D},Q_{{\cal D},0},\Delta_{\cal D},\Gamma_{\cal D} \rangle$ and the accepting simulation with update parity game ${\cal P}_{{\cal S}, {\cal D}}$.
  Assume that for all $q_{\cal S} \in Q_{{\cal S},0}$ we have $q_{\cal D} \in Q_{{\cal D},0}$ such that $(q_{\cal S}, q_{\cal D})$ is a winning state in the parity game.
  Then ${\cal D}$ accepting transition with update simulates ${\cal S}$.
  Otherwise, ${\cal D}$ does not simulate ${\cal S}$.
\end{lemma}
Again, Lemma~\ref{lem:accupdatesimdef} can be used in the same way as Lemma~\ref{lem:accsimdef}.


\end{document}